\documentclass[a4paper,fleqn]{cas-dc}
\usepackage[numbers]{natbib}

\usepackage{amsmath,amssymb,amsfonts}
\usepackage{mathtools}
\usepackage{bm}
\usepackage{amsthm}

\usepackage{graphicx}
\usepackage{subcaption}
\usepackage{booktabs}
\usepackage{threeparttable}
\usepackage{tabularx}
\usepackage{multirow}
\usepackage{longtable}
\usepackage{tablefootnote}
\usepackage{float}
\usepackage{lscape}
\usepackage{xcolor,colortbl}
\usepackage{wrapfig}
\usepackage{stfloats}
\usepackage{placeins}
\usepackage{algorithm}
\usepackage[noend]{algpseudocode}
\usepackage{algcompatible}

\usepackage{textcomp}
\usepackage{xspace}
\usepackage{pifont}
\usepackage{calc}
\usepackage{fancyhdr}
\usepackage[shortlabels]{enumitem}
\usepackage{url}
\usepackage{varioref}
\usepackage{tikz}
\usepackage{pgfgantt}
\usepackage[normalem]{ulem}  
\usepackage{lipsum}          
\usepackage{ulem}
\usepackage{xcolor}
\usepackage{times}

\newcommand{\algrule}[1][.2pt]{\par\vskip.5\baselineskip\hrule height #1\par\vskip.5\baselineskip}
{\bfseries}{\rmfamily}
{\bfseries}{\rmfamily}
{\bfseries}{\rmfamily}
\newtheorem{definition}{Definition}

\newcommand{\as}{\ensuremath {\leftarrow}{\xspace}}
\newtheorem{theorem}{Theorem}

\newcommand{\query}{\ensuremath {\texttt{Query}{\xspace}}}
\newcommand{\respond}{\ensuremath {\texttt{Respond}{\xspace}}}
\newcommand{\request}{\ensuremath {\texttt{Request}{\xspace}}}

\newcommand{\client}{\ensuremath {\texttt{Client}{\xspace}}}
\newcommand{\ue}{\ensuremath {\texttt{UE}{\xspace}}}

\newcommand{\server}{\ensuremath {\texttt{Server}{\xspace}}}
\newcommand{\service}{\ensuremath {\texttt{Service}{\xspace}}}

\newcommand{\cmark}{\ding{51}}%
\newcommand{\xmark}{\ding{55}}%

\newcommand{\slap}{\ensuremath {\texttt{SLAPX}{\xspace}}}
\newcommand{\pol}{\ensuremath {\texttt{PoL}{\xspace}}}

\newcommand{\ND}{\ensuremath {\texttt{ND}{\xspace}}}
\newcommand{\AP}{\ensuremath {\texttt{AP}{\xspace}}}
\newcommand{\dac}{\ensuremath {\texttt{DAC}{\xspace}}}

\newcommand{\dbp}{\ensuremath {\texttt{DBP}{\xspace}}}
\newcommand{\sym}{\ensuremath {\texttt{Sym}{\xspace}}}

\newcommand{\PSD}{\ensuremath {\texttt{PSD}{\xspace}}}

\newcommand{\keygen}{\ensuremath {\texttt{KeyGen}{\xspace}}}
\newcommand{\aka}{\ensuremath {\texttt{AKA}{\xspace}}}

\newcommand{\rlrs}{\ensuremath {\texttt{RLRS}{\xspace}}}
\newcommand{\setup}{\ensuremath {\texttt{Setup}{\xspace}}}
\newcommand{\extract}{\ensuremath {\texttt{Extract}{\xspace}}}

\newcommand{\sgn}{\ensuremath {\texttt{SGN}{\xspace}}}
\newcommand{\sign}{\ensuremath {\texttt{Sign}{\xspace}}}
\newcommand{\verify}{\ensuremath {\texttt{Verify}{\xspace}}}
\newcommand{\link}{\ensuremath {\texttt{Link}{\xspace}}}
\newcommand{\revoke}{\ensuremath {\texttt{Revoke}{\xspace}}}

\newcommand{\vdf}{\ensuremath {\texttt{VDF}{\xspace}}}
\newcommand{\eval}{\ensuremath {\texttt{Eval}{\xspace}}}
\newcommand{\nymgen}{\ensuremath {\texttt{NymGen}{\xspace}}}
\newcommand{\createcred}{\ensuremath {\texttt{CreateCred}{\xspace}}}
\newcommand{\issuecred}{\ensuremath {\texttt{IssueCred}{\xspace}}}
\newcommand{\credprove}{\ensuremath {\texttt{CredProve}{\xspace}}}
\newcommand{\receivecred}{\ensuremath {\texttt{ReceiveCred}{\xspace}}}
\newcommand{\credverify}{\ensuremath {\texttt{CredVerify}{\xspace}}}
\newcommand{\getcred}{\ensuremath {\texttt{GetCred}{\xspace}}}
\newcommand{\proxverify}{\ensuremath {\texttt{ProxVerify}{\xspace}}}

\newcommand{\psd}{\ensuremath {\texttt{PSD}}}

\newcommand{\ap}{\ensuremath {\texttt{AP}}}
\newcommand{\fcc}{\ensuremath {\texttt{FCC}}}
\newcommand{\nd}{\ensuremath {\texttt{ND}}}
\newcommand{\id}{\ensuremath {\texttt{ID}}}

\newcommand{\rtt}{\ensuremath {\texttt{RTT}}}

\newcommand{\msk}{\ensuremath {\texttt{msk}}}
\newcommand{\vk}{\ensuremath {\texttt{vk}}}

\newcommand{\sk}{\ensuremath {\texttt{sk}}}
\newcommand{\dk}{\ensuremath {\texttt{dk}}}
\newcommand{\pk}{\ensuremath {\texttt{PK}}}
\newcommand{\db}{\ensuremath {\texttt{DB}}}
\newcommand{\aux}{\ensuremath {\texttt{aux}}}
\newcommand{\nym}{\ensuremath {\texttt{nym}}}
\newcommand{\cred}{\ensuremath {\texttt{cred}}}
\newcommand{\ts}{\ensuremath {\texttt{TS}}}
\newcommand{\tv}{\ensuremath {\texttt{TV}}}
\newcommand{\rss}{\ensuremath {\texttt{RSS}}}

\newcommand{\squarenum}[1]{%
  \tikz[baseline=(char.base)]\node[
    shape=rectangle,
    draw=blue,
    fill=black,
    text=white,
    inner sep=2pt,
    rounded corners=2pt,
    scale=0.8
  ] (char) {\sffamily\bfseries\small #1};%
}

\def\tsc#1{\csdef{#1}{\textsc{\lowercase{#1}}\xspace}}
\tsc{WGM}
\tsc{QE}

\begin{document}
\let\WriteBookmarks\relax
\def\floatpagepagefraction{1}
\def\textpagefraction{.001}

\shorttitle{Privacy-Preserving and Secure DB-Driven Spectrum Sharing}    

\shortauthors{Darzi et al.}  

\title [mode = title]{Privacy-Preserving and Secure Spectrum Sharing for Database-Driven Cognitive Radio Networks}


\author[1]{Saleh Darzi}
\cormark[1]
\ead{salehdarzi@usf.edu}
\affiliation[1]{organization={Bellini College of Artificial Intelligence, Cybersecurity and Computing, University of South Florida},
            city={Tampa},
            postcode={33620}, 
            state={Florida},
            country={USA}}

\author[2,3]{Gökcan Cantali}
\fnmark[2]
\ead{canl@zhaw.ch}
\affiliation[2]{organization={Zurich University of Applied Sciences (ZHAW)},
            city={Winterthur},
            postcode={8401}, 
            country={Switzerland}}
\affiliation[3]{organization={Department of Informatics, University of Zurich (UZH)},
            city={Zurich},
            postcode={8450}, 
            country={Switzerland}}

\author[1]{Attila Altay Yavuz}
\fnmark[3]
\ead{attilaayavuz@usf.edu}

\author[2]{Gürkan Gür}
\fnmark[4]
\ead{gurkan.gur@zhaw.ch}


\begin{abstract}
Database-driven cognitive radio networks (DB-CRNs) enable dynamic spectrum sharing through geolocation databases but introduce critical security and privacy challenges, including mandatory location disclosure, susceptibility to location spoofing, and denial-of-service (DoS) attacks on centralized services. Existing approaches address these issues in isolation and lack a unified, regulation-compliant solution under realistic adversarial conditions.  
In this work, we present a unified security framework for DB-CRNs that simultaneously provides location privacy, user anonymity, verifiable location, and DoS resilience. Our framework, denoted as $\slap$, enables privacy-preserving spectrum queries using delegatable anonymous credentials, supports adaptive location verification without revealing precise user location, and mitigates DoS attacks through verifiable delay functions (VDFs) combined with RLRS-based rate limiting. Extensive cryptographic benchmarking and network simulations demonstrate that $\slap$ achieves significantly lower latency and communication overhead than existing solutions while effectively resisting location spoofing and DoS attacks. These results show that $\slap$ is practical and well-suited for secure next-generation DB-CRN deployments. 
\end{abstract}



\begin{keywords}
Database-Driven Cognitive Radio Networks \sep Location Privacy \sep Anonymous Credentials \sep Location Proof \sep Counter-DoS 
\end{keywords}

\maketitle

\section{Introduction}
\label{sec:introduction}
The rapid growth of wireless technologies, driven by mobile devices, Internet-of-Things (IoT) deployments, and emerging spectrum-hungry applications, has intensified pressure on traditionally static spectrum allocation policies~\cite{agarwal2022survey,chakraborty2023capow,9321158}. To address spectrum underutilization, database-driven cognitive radio networks (DB-CRNs) and dynamic spectrum sharing have been introduced, enabling Secondary Users (SUs) to opportunistically access licensed spectrum owned by Primary Users (PUs) through regulatory-approved geolocation databases (DBs). A prominent example is the U.S. Citizens Broadband Radio Service (CBRS), which operates in the $3.5$~GHz band originally allocated to federal and satellite services~\cite{agarwal2022survey, grissa2021anonymous}. Moreover, dynamic spectrum sharing is expected to be a foundational element in future networks such as 6G systems~\cite{PATIL2024110697}.

While this paradigm significantly improves spectrum efficiency, it also introduces a wide range of security and privacy challenges. In particular, regulatory authorities (e.g., FCC and ITU) require users to continuously report sensitive operational information, including precise location, device identity, and transmission parameters, to centralized DBs. This requirement raises serious concerns regarding user privacy, anonymity, and long-term profiling~\cite{jasim2021cognitive, grissa2016efficient}. Moreover, the location-centric design of DB-CRNs makes them inherently vulnerable to location spoofing attacks, in which malicious or compromised users falsify their positions to obtain unauthorized spectrum access~\cite{xin2016privacy}. At the same time, the centralized reliance on spectrum databases, the broadcast nature of wireless communication, and the proliferation of low-cost devices further expose DB-CRNs to denial-of-service (DoS) attacks that can disrupt spectrum coordination and degrade network availability~\cite{darzi2024privacy}. 

Existing solutions for DB-CRNs largely address individual challenges in isolation and do not provide a unified, regulation-compliant framework that jointly ensures location privacy, anonymity, verifiable location, and resilience against DoS attacks under realistic adversarial conditions. 
In this work, to address those research gaps, our contributions are threefold:\vspace{-2mm}

\begin{enumerate}[leftmargin=*]
\item We develop and evaluate a novel and efficient framework, namely $\slap$, that leverages advanced cryptographic primitives to simultaneously address multiple, often conflicting, objectives in DB-CRNs. The proposed design achieves user location privacy and anonymity during spectrum querying and access while complying with enforced regulations. This paper extends our \texttt{SLAP} framework\footnote{The preliminary version of our approach, namely SLAP, was presented and published at the \textit{IEEE SVCC 2025 Conference}~\cite{darzi2025slap}.}, incorporating substantial new material and the following key enhancements and additional features. Designed for realistic adversarial settings, $\slap$ remains robust against location spoofing and DoS attacks by malicious or compromised users. The key properties of $\slap$ are outlined as follows:
\textit{(i) Efficient Location Privacy-Preserving and Anonymous Spectrum Query:} $\slap$ provides strong anonymity and location privacy throughout the spectrum access workflow while remaining fully compliant with protocol and regulatory requirements. 
\textit{(ii) Adaptive Location Proof and Spoofing Resistance for DB-CRN:} $\slap$ comprises an adaptive location proof algorithm with dual scenario support and architectural flexibility that operates under realistic network assumptions with malicious
or compromised users. 
\textit{(iii) DoS Mitigation with DB-CRN Architecture Compliance:} $\slap$ provides comprehensive DoS mitigation through a proactive defense mechanism that combines Verifiable Delay Functions (VDFs) with rate limiting enforced via the linkability properties of Revocable-iff-linked Linkable Ring Signatures (RLRSs). \vspace{-2mm}

\item We present a system model that integrates $\slap$ into a DB-CRN architecture comprising user equipment (UEs), access networks, spectrum databases, and regulatory authorities. This model provides a practical baseline for implementing and benchmarking DB-CRN security mechanisms, and serves as a foundation for evaluating future enhancements, including quantum-safe extensions. \vspace{-2mm}

\item We provide extensive simulation-based experiments and numerical analysis, which provide a blueprint for realistic  evaluation and benchmarking for further efforts regarding the investigated research questions on privacy-preserving and secure spectrum sharing.
\end{enumerate}

Our paper is organized as follows: in the next section, we present the technical preliminaries, system model, and cryptographic building blocks for our proposed $\slap$ framework. In Section~\ref{sec:threatmodel}, we construct threat and security models and attack scenarios relevant to $\slap$. Subsequently, we describe $\slap$'s core operations and the detailed algorithmic specification of the framework in Section~\ref{sec:scheme}. Then, in Section~\ref{sec:performance}, a comprehensive performance evaluation with simulations and numerical experiments is presented. We discuss the related work to position our contributions in the current body of work and concisely highlight the research gaps addressed by our work to elaborate on our contributions in Section~\ref{sec:discussion}. 
Finally, we conclude with key takeaways and future work.

\section{Preliminaries and Building Blocks}
\label{sec:prelim}
This section presents the notations, network architecture, and cryptographic building blocks of our framework. 

\textbf{Notations:}
The operators $||$, $|x|$, and $\{0,1\}^k$ denote string concatenation, the bit-length of a value $x$, and the set of $k$-bit binary strings, respectively, while $\oplus$ represents the bitwise XOR operation. The set of natural numbers is denoted by $\mathbb{N}$. For a sequence $\{x_i\}_{i=1}^{\ell}$, we write $(x_1,\ldots,x_\ell)$, and $\xleftarrow{\$}\mathcal{S}$ denotes uniform random sampling from a set $\mathcal{S}$. Let $\mathbb{G}_1$, $\mathbb{G}_2$, and $\mathbb{G}_T$ be cyclic groups of prime order $p$, and let ${e:\mathbb{G}_1 \times \mathbb{G}_2 \rightarrow \mathbb{G}_T}$ denote a bilinear pairing satisfying bilinearity and non-degeneracy. The $i$-th element of a vector $\mathbf{m}$ is denoted by $\mathbf{m}[i]$, and $h(\cdot)$ represents a cryptographically secure hash function. The security parameter is denoted by $\lambda$ and $p$ denotes a large prime. Finally, $\sk$, $\pk$, and $\id$ denote a secret key, a public key, and a user identifier (e.g., a MAC address), respectively, while $\sgn.\sign(\cdot)$ and $\sgn.\verify(\cdot)$ represent standard digital signature generation and verification algorithms.


\subsection{System Architecture: DB-CRN} \label{subsec:systemarch}

\textbf{Network Components:} The DB-CRN is a spectrum-sharing architecture in which unlicensed devices query trusted geolocation databases to obtain channel availability at their current location~\cite{PATIL2024110697}. This model was deployed in real FCC TV White Space (TVWS) systems and adopted in the Microsoft and Google database trials~\cite{caleffi2014database, das2015rfc}. The DB-CRN architecture consists of five key entities~\cite{order2023federal}: \vspace{-2mm}

\begin{itemize}[leftmargin=*] 
\item \textbf{\textit{Spectrum Regulator}:} It is the regulatory authority that defines technical rules for DB-CRN operation, like incumbent protection requirements, device reporting and registration procedures, and transmit parameters~\cite{agarwal2022survey}. This entity authorizes and oversees geolocation database administrators and provides authoritative incumbent datasets and regulatory constraints. Throughout this manuscript, the FCC is used interchangeably to denote this spectrum regulatory authority for convenience. \vspace{-2mm}
    

\item \textbf{\textit{Servers}:} They represent external network or cloud services (e.g., CRN, cloud, edge nodes) that users access after obtaining spectrum authorization~\cite{grissa2016efficient}. They do not participate in spectrum assignment; instead, they simply carry user traffic over the channels approved by the $\psd$. \vspace{-2mm}


\item \textbf{\textit{Private Spectrum Databases (PSDs)}:} PSDs are FCC-certified, cloud-based geolocation databases (e.g., Google, Microsoft, Spectrum Bridge) responsible for maintaining incumbent information, propagation constraints, and channel availability maps~\cite{chen2015protocol}. PSDs process spectrum-availability queries submitted by users and return permissible channels and power limits, synchronizing regularly according to FCC mandates~\cite{grissa2021anonymous}. \vspace{-2mm}


\item \textbf{\textit{User Equipments (UEs)}:} UEs are unlicensed devices of users (e.g., smartphones, sensors, and vehicular radios) that seek opportunistic access to licensed spectrum bands~\cite{grissa2016efficient}. Acting as SUs, they determine their location, submit channel-availability queries to PSDs, and operate strictly within the constraints specified in PSD responses to avoid interfering with PUs. Throughout this paper, the terms client and UE are used interchangeably. \vspace{-2mm}


\item \textbf{\textit{Access Points (APs)}:} APs consist of WiFi access points or cellular base stations (LTE/5G macro or small cells) that provide the IP connectivity required for UEs to reach PSDs and servers over the Internet~\cite{darzi2025slap}. APs do not manage spectrum availability but offer backhaul connectivity and, in some advanced architectures, assist with localization or timing support. \vspace{-2mm}

\end{itemize}

\textbf{Communication Flow and Protocol Stack:} In a DB-CRN, UEs communicate with PSDs over standard Internet paths to obtain spectrum availability, while servers and APs provide network connectivity and application-layer services. Communication exchanges follow an IP-based client–server model and use secure, lightweight protocols consistent with real FCC TVWS deployments and the IETF PAWS specification~\cite{liang2011cognitive}. The communication workflow consists of three primary interactions: \vspace{-2mm}

\begin{itemize}[leftmargin=*]
\item \textit{UE-to-AP Communication (Access Connectivity):} UEs attach to a nearby AP to obtain Internet connectivity. This step uses standard link-layer and access-network protocols, such as IEEE~802.11 or 3GPP RRC/NAS for LTE/5G~\cite{thippeswamy2016physical}. APs forward UE traffic to the Internet and do not participate in the decision-making process~\cite{liang2011cognitive}. \vspace{-2mm}

\item \textit{$\ue$-to-$\psd$ Communication (Spectrum Querying):} UEs submit spectrum-availability queries directly to PSDs using HTTPS over TLS to ensure confidentiality and integrity of location and device parameters. The request and response formats follow the JSON-based geolocation DB access protocol standardized by the IETF PAWS specification (RFC-7545)~\cite{liang2011cognitive}. A typical UE-to-PSD exchange includes (i) a registration or device-descriptor message, (ii) a location-based $\textit{AVAIL\_CHANNELS}$ query, and (iii) a PSD response containing permissible channels, power limits, and validity intervals. PSDs process each request using incumbent datasets and FCC rules before sending their response~\cite{thippeswamy2016physical}. \vspace{-2mm}

\item \textit{$\ue$-to-Server Communication (Application Traffic):} After obtaining spectrum authorization, UEs communicate with application or cloud Servers using standard IP protocols. These Servers do not participate in spectrum allocation; they merely receive traffic over the channels authorized by the PSD. Servers may verify spectrum-access credentials or tokens in enhanced architectures, but this is not required in the baseline DB-CRN model.\vspace{-2mm}
\end{itemize}

\begin{figure}[ht]
    \includegraphics[width=\linewidth]{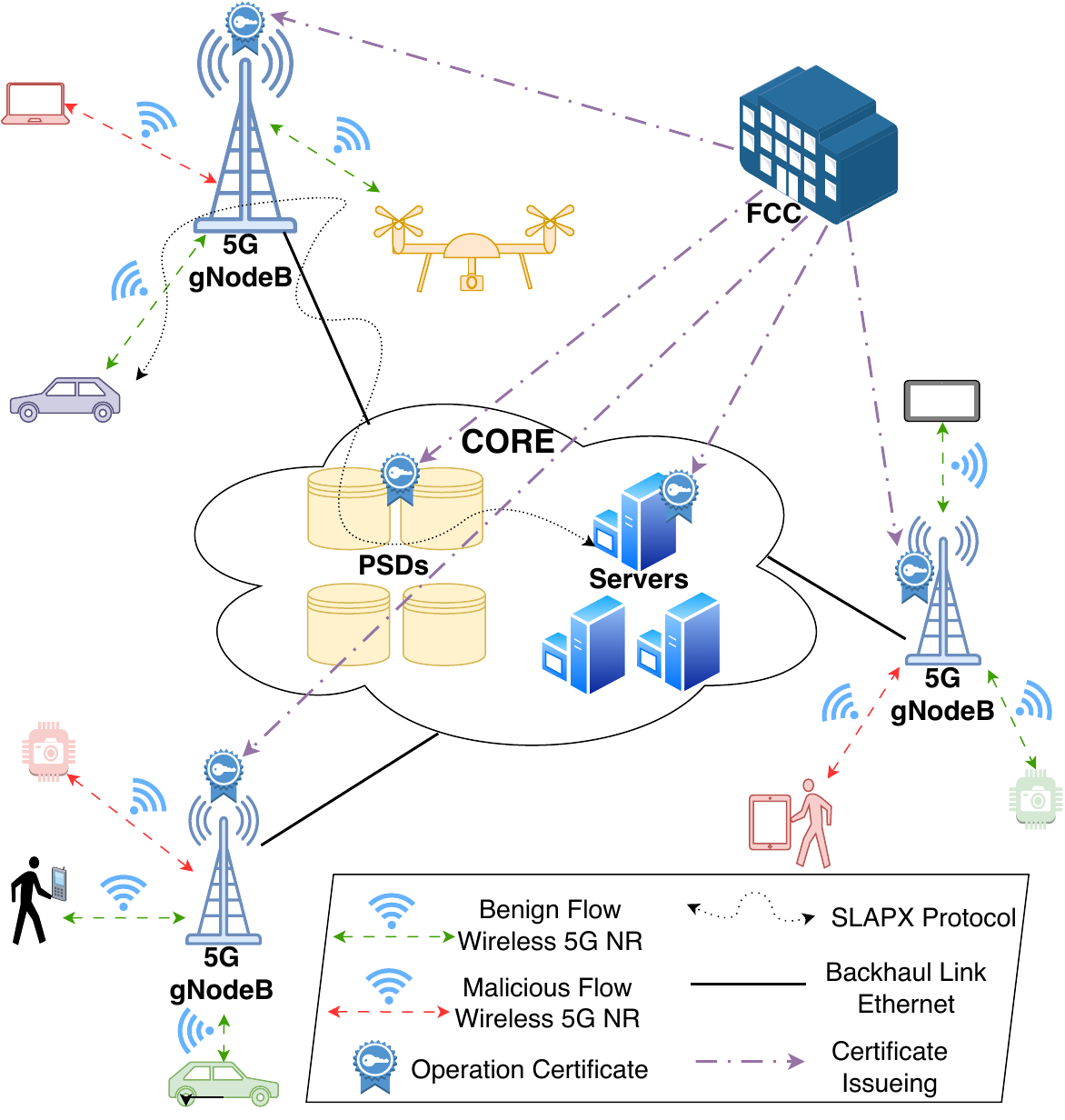}
    \caption{Overview of DB-CRN system architecture}
    \label{fig:SLAP-System-Architecture}\vspace{-3mm}
\end{figure}

\noindent The communication stack follows a standard Internet protocol hierarchy, ensuring seamless interoperability with existing devices and networks. This design enables secure, fine-grained, and location-aware spectrum access without requiring any modifications to underlying access infrastructures or access points. The overall system architecture of DB-CRN is illustrated in Figure~\ref{fig:SLAP-System-Architecture}.

\subsection{Cryptographic Building Blocks}\label{subsec:buildingblocks}
$\slap$ harnesses various cryptographic components to achieve its privacy-preserving and secure operation goals: 


\noindent~$\bullet$ \textbf{Delegatable Anonymous Credentials (${DAC}$):} We use an attribute-based DAC scheme~\cite{mir2023practical} that supports multi-level attribute issuance, controlled delegation, and anonymity. The main algorithms are summarized below (see~\cite{mir2023practical}).

\begin{definition} \label{def:DAC}
\normalfont
A Delegatable Anonymous Credential ($\dac$) scheme is defined by the following algorithms: \vspace{-2mm}
\begin{itemize}[leftmargin=*]
\setlength{\itemsep}{1pt} 
    \item[-] ${(pp, \sk_{RI}, \pk_{RI}) \as \setup(1^\lambda, 1^t, 1^\eta)}$: Given the security parameter $\lambda$, a bound $t$ on the committed attribute set size, and a delegation depth parameter ${\eta>1}$, it initializes the system by producing public parameters $pp$ and generating the root issuer’s key pairs ${(\sk_{RI}, \pk_{RI})}$ for each level in $\eta$. $pp$ is an implicit input to other algorithms. \vspace{-1mm}

    
    \item[-] ${(\pk, \sk) \as \keygen(pp)}$: Generates the user’s key pair $(\pk, \sk)$, which $\pk$ acts as the user’s initial pseudonym. \vspace{-1mm}
    
    \item[-] ${(\nym, \aux) \as \nymgen(\pk)}$: From the user's $\pk$, outputs a fresh pseudonym $\nym$ and auxiliary randomness $\aux$ to maintain unlinkability in subsequent interactions. \vspace{-1mm}
    

    \item[-] $\createcred(L', A, \sk_{RI}) \leftrightarrow$ $\getcred(\pk_{RI}, \sk_u, A)$ $\rightarrow (\cred, (\overrightarrow{C},\overrightarrow{O}), \dk_{L'})$: An interactive protocol between the root issuer (RI) and a user identified by $\nym_u$. For the attribute set $A$, the RI forms a set commitment $C$ and signs it using commitment to create a delegatable credential at $\pk_{RI}$. The user receives the credential $\cred$, the opening information $\overrightarrow{O}$, and a delegation key $\dk_{L'}$ for level $L'$. \vspace{-1mm}

    \item[-] ${\issuecred(\pk_{RI}, \dk_{L'}, \sk_u, \cred_u, A_l, L'') \leftrightarrow \texttt{Recei}}$ ${\texttt{veCred}(\pk_{RI}, \sk_r, A_l) \rightarrow (\cred_r, \dk'_{L''})}$: It enables controlled delegation from a delegator $\nym_i$ to a recipient $\nym_r$. Using their secret key $\sk_i$, public key of the RI $\pk_{RI}$, credential $\cred_i$, delegation key $\dk_{L'}$, and attribute extension $A_l$, the delegator generates a new credential $\cred_r$ for the recipient. The delegatee, using the secret key $\sk_r$, obtains the credential $\cred_r$ with the updated attribute set $A' = (A, A_l)$ and delegation depth $L'' \leq L'$, along with a new delegation key $\dk'_{L''}$ that supports further delegation if allowed. \vspace{-1mm}

    \item[-] ${\credprove(\pk_{RI}, \sk_{p}, \nym_p, \aux_p, \cred_p, D) \leftrightarrow \texttt{Cre}}$ ${\texttt{dVerify}(\pk_{RI}, \nym_p, D) \rightarrow \{0, 1\}}$: A Zero-Knowledge (ZK) presentation protocol that allows a user to anonymously prove possession of a valid credential. The prover constructs a proof using their secret key $\sk_p$, pseudonym $\nym_p$, randomness $\aux_p$, and credential $\cred_p$, revealing only the disclosed attributes $D$. The verifier checks the proof under $\pk_{RI}$ and outputs $1$ if the credential is valid, and $0$ otherwise.
\end{itemize}
\end{definition}

\noindent~$\bullet$ \textbf{Distance Bounding Protocol (DBP):}
A DBP is a two-party protocol that verifies the physical proximity of two entities by measuring round-trip delays in a fast challenge–response exchange. We employ a PK-DBP~\cite{kilincc2016efficient} based on a one-pass authenticated key-agreement protocol using nonce Diffie–Hellman to derive a session key between the prover and verifier. 
This key is then used within the OTDB symmetric-DBP~\cite{vaudenay2015private}. 

\begin{definition} \label{def:DBP}
\normalfont
A public key DBP is defined as follows: \vspace{-1mm}
\begin{itemize}[leftmargin=*]
\setlength{\itemsep}{1pt} 
    \item[-] ${(\sk, \pk) \as \keygen(1^\lambda)}$: Generates a public/secret key pair for security parameter $\lambda$.\vspace{-1mm}
    
    \item[-] ${ss \as \aka(\sk, \pk, \pk')}$: The prover and verifier derive a shared session key using their own key pairs and the other party’s public key ($\pk'$).\vspace{-1mm}
    
    \item[-] ${\{0,1\} \as \sym\dbp(ss, \textit{th})}$: An interactive proximity-test protocol using session key $ss$ and distance threshold $th$, consisting of: 
    {\em (i)} \texttt{Initialization:} The verifier samples $m \in \{0,1\}^{2n}$ and sends it to the prover, who computes $a = ss \oplus m$. 
    {\em (ii)} \texttt{Rapid Bit Exchange:} For $i=1,2,...,n$, the verifier sends a challenge $c_i \in \{0,1\}$ and the prover responds with $r_i = a_{2i+c_i-1}$, while the verifier records the round-trip time $\rtt_i$.    
    {\em (iii)} \texttt{Authentication:} Using $a = ss \oplus m$, the measured delays, and the distance bound, the verifier checks ${\rtt_i \leq 2\times \textit{th}}$ and $r_i = a_{2i+c_i-1}$. It outputs $1$ if the prover is within the threshold; otherwise $0$.
\end{itemize}
\end{definition}

\noindent~$\bullet$ \textbf{Revocable-iff-linked Linkable Ring Signature (RLRS):} Ring signatures allow a user to anonymously sign on behalf of a group. An RLRS augments this by embedding an event identifier into each signature, enabling detection of multiple signatures produced by the same signer within that context and supporting controlled anonymity revocation when necessary~\cite{au2013secure}. \vspace{-2mm} 

\begin{definition} \label{def:RLRS}
\normalfont
A Revocable-iff-linked Linkable Ring Signature ($\rlrs$) scheme is a 6-tuple algorithm as shown below: \vspace{-4mm}
\begin{itemize}[leftmargin=*]
\item[-] ${(\msk, pp_\rlrs) \as \rlrs.\setup(1^\lambda)}$: Given $\lambda$, generate the master secret key $\msk$ and the public params $pp_\rlrs$. \vspace{-3mm}

    \item[-] ${(\sk_u) \as \rlrs.\extract(pp, \id_u)}$: An interactive key-extraction protocol that, given $pp_\rlrs$ and a user identifier $\id_u \in \{0,1\}^*$, outputs the user secret key $\sk_u$. \vspace{-3mm}

    
    \item[-] ${(\sigma, \tau) \as \rlrs.\sign(\sk_u, m, \mathcal{L}_\id, e_\id)}$: To sign a message $m$ under identifier $e_\id$ with ring $\mathcal{L}_\id = \{\id_1, \dots, \id_n\}$ ($n \le t_{\max}$), the signer uses a non-interactive ZK proof to produce a signature $\sigma$ and link tag $\tau$. \vspace{-3mm}

    \item[-] ${\{0,1\} \as \rlrs.\verify(\mathcal{L}_\id, m, e_\id, \sigma, \tau)}:$ Given the ring $\mathcal{L}_\id$, message $m$, identifier $e_\id$, and signature-tag pair $(\sigma,\tau)$, output $1$ if the signature is valid, or else $0$. \vspace{-3mm}

    \item[-] ${\{0,1\} \as \rlrs.\link(\mathcal{L}_\id, e_\id, m, \sigma, \tau, m', \sigma', \tau')}:$ Determine if two signature-tag pairs $(\sigma,\tau)$ and $(\sigma',\tau')$ on messages $m$ and $m'$ (under event $e_\id$ and ring $\mathcal{L}_\id$) were produced by the same signer. Return $1$ if linked, or else $0$. \vspace{-3mm}
    
    \item[-] ${(\perp, \id_u) \as \rlrs.\revoke()}$: Given two linked signatures from events $e_\id$ and rings $\mathcal{L}_\id$ and $\mathcal{L}'_\id$, output the signer’s identity $\id_u \in \mathcal{L}_\id \cap \mathcal{L}'_\id$ if the signatures originate from the same user; output $\perp$ otherwise. 
\end{itemize}
\end{definition}

        
        
        


\noindent~$\bullet$ \textbf{Verifiable Delay Function (VDF):} A VDF~\cite{boneh2018verifiable} enforces a prescribed delay by requiring sequential computation, while enabling a fast verification of the result. We use an RSA-based VDF~\cite{wesolowski2020efficient}, which provides compact proofs and constant-time verification. \vspace{-2mm}

\begin{definition} \label{def:VDF}
\normalfont
A Verifiable Delay Function (\vdf) scheme is a 3-tuple algorithm as shown below: \vspace{-2mm}
\begin{itemize}[leftmargin=*]
	\item[-] ${(N, H, H_p) \as \vdf.\setup(1^\lambda, \kappa)}$: Given $\lambda$ and difficulty level $\kappa$, generate an RSA modulus $N$ of size $\lambda$, a hash function $H: \{0,1\}^*\rightarrow \{0,1\}^{2\kappa}$, and a prime-deriving hash $H_p(m) = \texttt{next-prime}(H(m))$, which returns the closest prime number greater than or equal to $H(m)$.  \vspace{-2mm}
    
    \item[-] ${(\ell, \pi) \as \vdf.\eval(m, \tau)}$: On input $m\in\{0,1\}^*$ and delay parameter $\tau \in \mathbb{N}$, compute $x = H(m)$ and the sequentially run $y = x^{2^\tau}\mod N$. Then compute the proof using $\ell = H_p(x+y)$ and $\pi = x^{\lfloor 2^\tau/\ell \rfloor}$. \vspace{-2mm}
    
    \item[-] ${\{0,1\} \as \vdf.\verify(x, \tau, \ell, \pi)}$: Given $x$, $\tau$, $\ell$, $\pi$, compute ${x=H(m)}$, ${r=2^\tau \mod \ell}$, and $y = \pi^\ell \times x^r \mod N$. The verifier outputs $1$ if $\ell = H_p(x+y)$, or else $0$. 
\end{itemize}
\end{definition}

\section{Threat and  Security Models and Attack Scenarios}
\label{sec:threatmodel}
This section outlines (i) the threat model and scope of our solution, defining the adversary, its capabilities, and the considered attack vectors; (ii) the security objectives and security model that underpin the formal security proof of $\slap$ as a primitive; and (iii) representative attack scenarios that motivate and guide our network-level simulations when our proposed cryptographic techniques are integrated into the network architecture.

\subsection{Threat Model and Scope}\label{subsec:threatmodel} 
Our threat model considers a probabilistic polynomial-time (PPT) adversary with control over the wireless communication channel. This adversary seeks to violate user privacy by targeting confidentiality, identity protection, message authenticity, and integrity during spectrum access, while also attempting to disrupt service through DoS attacks against PSDs and backend servers. In addition, the adversary may exploit the system by submitting fraudulent or manipulated location information. PSDs, and service servers are assumed to be honest-but-curious: they follow the prescribed protocols correctly but may attempt to infer sensitive information, such as user identity or precise location. Under this model, the adversary can mount three attack vectors commonly observed in DB-CRNs, which are captured in our threat assumptions and described in detail below: \vspace{-2mm}

\begin{itemize}[leftmargin=*]
    \item \textit{Client Privacy and Anonymity}:
    Adversaries such as PSDs, service servers, or external observers may attempt to deduce a user’s real-world identity, precise location coordinates, or device attributes by analyzing protocol messages exchanged during spectrum queries and access requests. These efforts can include passive monitoring, active interaction, or correlating metadata to uncover sensitive user-specific information. \vspace{-2mm}
    
    \item \textit{Location Spoofing Attacks}: Malicious users may attempt to bypass access restrictions by reporting incorrect location information to obtain spectrum resources or services beyond their authorized area. Such misuse can also arise from compromised or spoofed devices. Relevant threats include falsified coordinates, replayed location proofs, collusion with adversarial entities, distance fraud, relay or mafia fraud, and distance hijacking~\cite{nguyen2019spoofing}. \vspace{-2mm}

    \item \textit{Denial-of-Service Attacks}: Malicious and compromised users or external attackers may attempt to disrupt DB-CRN operations by overwhelming PSDs or service servers during spectrum queries or service requests. Such attacks can involve flooding the system with invalid, replayed, or resource-intensive requests to exhaust computational capacity, increase response latency, or deny service to legitimate users. These threats are particularly harmful in distributed and delay-sensitive spectrum access environments, where timely and reliable availability is essential.
\end{itemize}



\textbf{Scope of Our Solution:} The $\slap$ framework is designed to preserve user anonymity and location privacy during spectrum access. It protects SUs while they query spectrum availability, obtain cryptographic puzzles, and access network services, mitigating threats such as identity exposure, location spoofing, and DoS attacks. Our work concentrates exclusively on the spectrum query and access phases. Privacy risks arising during user registration or from spectrum usage patterns (e.g., utilization inference) fall outside the scope of this study. Although the framework is tailored for SUs, it can be extended to support PUs interacting with PSD servers; however, PU-related protections are not considered here. Moreover, $\slap$ does not address location leakage during ongoing spectrum use, user mobility or handover processes~\cite{gao2012location}, nor does it cover timing-based attacks, side-channel leakage, or signal-based localization methods such as triangulation~\cite{bahrak2014protecting}.


\subsection{Security Model}
Building on the system and threat models, $\slap$ is evaluated under the following security definitions:




\begin{definition}[Credential Unforgeability]\label{def:unforgeability}
A $\dac$ scheme $\cred$ is \emph{unforgeable} if no PPT adversary can cause an honest verifier to accept a credential proof for attributes it is not entitled to. Technically, $\Pr[
\credverify(\pk_{RI}, \nym^\star, D^\star)$ ${=1 \wedge D^\star \not\subseteq A_i] \leq negl(\lambda)}$, where $\mathcal{A}$ may obtain credentials via the $\createcred$ and $\issuecred$ protocols, $(\nym^\star, D^\star)$ is the adversary’s output in a showing attempt, and $A_i$ denotes any attribute set for which $\mathcal{A}$ has legitimately obtained a credential (including delegation).
\end{definition}

\begin{definition}[Credential Anonymity]\label{def:anonymity}
A $\dac$ scheme $\cred$ provides \emph{anonymity} if no PPT adversary can distinguish which of two legitimate credential holders produced a valid credential proof. Technically, $\Pr[\mathcal{A}^{\credprove(\pk_{RI})}(i_0,i_1)=$ ${1] \leq \tfrac{1}{2} + negl(\lambda)}$, where $\mathcal{A}$ is given two distinct user indices $(i_0,i_1)$ corresponding to honest users holding valid credentials, interacts with a verification oracle that executes $\credprove$ and $\credverify$ on behalf of one of the two users chosen uniformly at random, and attempts to guess which user generated the accepted proof.
\end{definition}

\begin{definition}[Credential Unlinkability]\label{def:unlinkability}
A $\dac$ scheme $\cred$ provides \emph{unlinkability} if no PPT adversary can determine whether two valid credential showings originate from the same credential holder. Technically, $\Pr[\mathcal{A}(\cred_0,\cred_1)$ ${=1] \leq \tfrac{1}{2} + negl(\lambda)}$, where $\cred_0$ and $\cred_1$ are two accepted protocol credentials generated via $\credprove$ and $\credverify$, and $\mathcal{A}$ is given oracle access to credential issuance and verification and attempts to distinguish whether $\cred_0$ and $\cred_1$ were produced by the same user or by two distinct users holding valid credentials.
\end{definition}

\begin{definition}[Correctness and Soundness of Location Verification]\label{def:locationverification}
A location verification scheme satisfies \emph{correctness} if any honest user physically located at $(l_x,l_y)$ can generate a proof of location $\pol$ that is accepted by the verifier with overwhelming probability, i.e., $\Pr[\texttt{Verify}(\pol(l_x,l_y)) =$ ${1] \ge 1-negl(\lambda)}$. The scheme satisfies \emph{soundness} if no PPT adversary $\mathcal{A}$ can produce a proof $\pol'$ that is accepted as valid for a location ${l' \neq l}$, except with negligible probability, unless $\mathcal{A}$ is physically present at $l'$, i.e., ${\Pr[\texttt{Verify}(\pol') = 1 \;\wedge\; \texttt{Loc}(\pol') = l'] \le negl(\lambda)}$. Va-lid proofs are cryptographically bound to their spatio-temporal context, ensuring non-transferability and replay resistance, and thereby preventing relay, distance-fraud, mafia-fraud, and location-hijacking attacks.
\end{definition}

\begin{definition}[Counter-DoS]\label{def:counterdos}
A system is \emph{DoS-resilient} if no PPT adversary $\mathcal{A}$ issuing at most $q(\lambda)$ requests can cause service unavailability—defined as delaying any honest request beyond a fixed bound $\Delta$—without incurring a proportional computational cost. Technically, $\Pr[
\texttt{Delay}(\mathcal{A}) >$ ${\Delta \;\wedge\; \texttt{Cost}(\mathcal{A}) < q(\lambda)\cdot\tau] \le negl(\lambda)}$, where $\tau$ denotes the enforced per-request cost determined by $\vdf$ evaluation, rate-limiting, and authenticated queries. 
\end{definition}

\subsection{Attack Scenarios}\label{subsec:attackscenarios} 
In addition to the primitive-level security analysis, we consider attack scenarios for DoS and location spoofing at the network-level. These scenarios motivate and guide the network-level security simulations of the $\slap$.



\begin{figure}[ht]
    \centering
    \begin{subfigure}{\linewidth}
        \centering
        \includegraphics[width=0.85\linewidth]{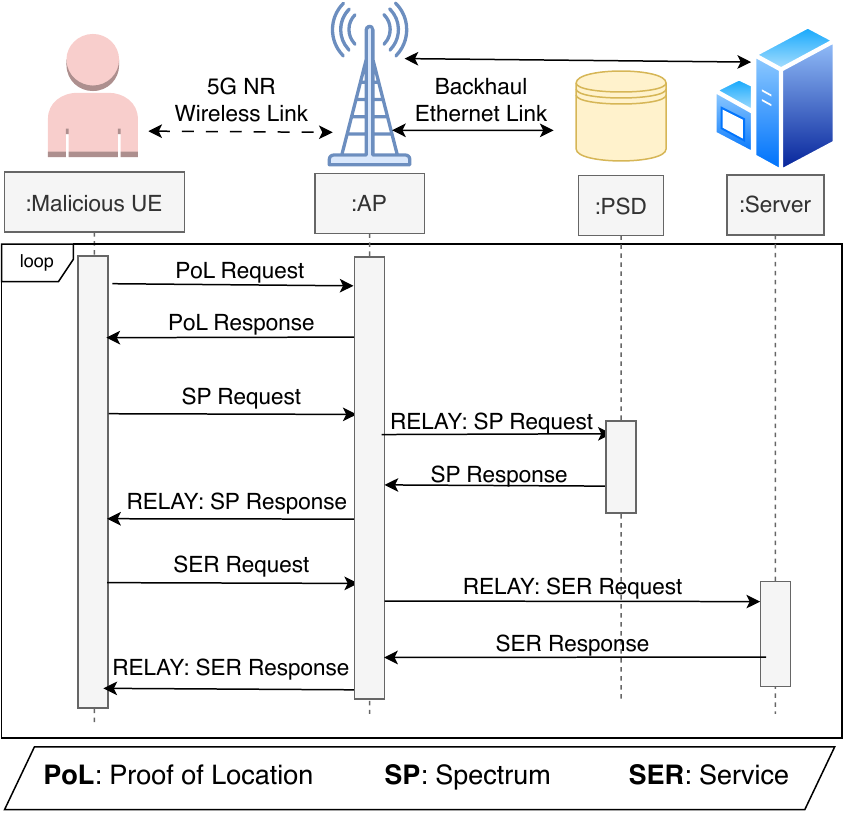}\vspace{-1mm}
        \caption{Flow of DoS Scenario 1: Full Protocol DoS}
        \label{fig:SLAP-DoS-flow-1}
    \end{subfigure}

    \begin{subfigure}{\linewidth}
        \centering
        \includegraphics[width=0.9\linewidth]{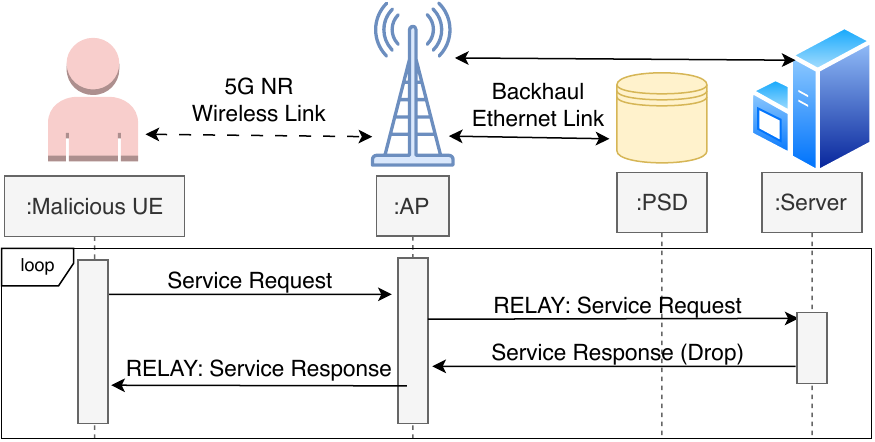}\vspace{-1mm}
        \caption{Flow of DoS Scenario 2: Computation-Bypassing DoS}
        \label{fig:SLAP-DoS-flow-2}
    \end{subfigure}

    \begin{subfigure}{\linewidth}
        \centering
        \includegraphics[width=0.90\linewidth]{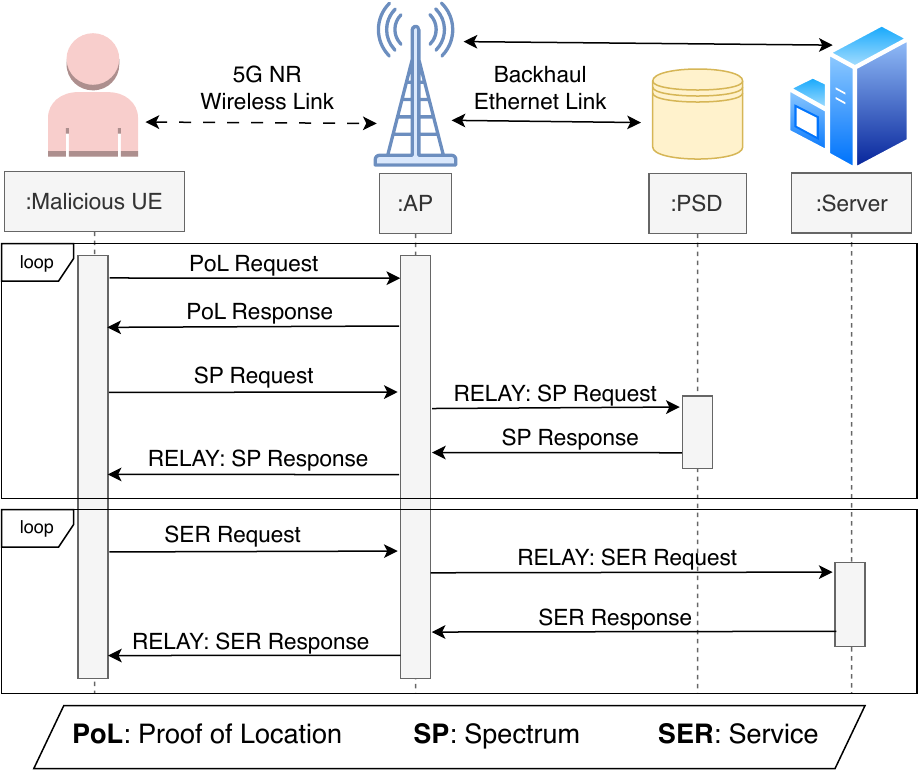}\vspace{-1mm}
        \caption{Flow of DoS Scenario 3: Proof Precomputation DoS}
        \label{fig:SLAP-DoS-flow-3}
    \end{subfigure}

    \caption{$\slap$ DoS attack scenarios 
    }
    \label{fig:SLAP-DoS-Modeling}\vspace{-4mm}
\end{figure}

\subsubsection{DoS Attack Scenarios}
We consider three representative DoS attack scenarios to capture protocol-compliant and protocol-abusive behaviors, as well as high-cost and low-cost adversarial strategies under realistic threat conditions. All scenarios are summarized in Figure~\ref{fig:SLAP-DoS-Modeling} and are analyzed in Section~\ref{subsec:simulation}. 

\noindent\textbf{1) Full Protocol DoS}: The malicious users repeatedly execute the $\slap$ algorithms with the goal of exhausting server-side computational and networking resources. While an honest user typically executes the protocol only once to obtain service access, the adversary continuously restarts the protocol to keep the servers busy. 
    
    
\noindent\textbf{2) Computation-Bypassing DoS}: The malicious users flood the  server with service requests without executing the prerequisite steps of$\slap$. As a result, these requests lack valid cryptographic proofs and are discarded upon verification. However, by bypassing expensive client-side computation, adversaries can generate a significantly higher request rate, enabling large-scale flooding attacks that stress server-side filtering and rate-limiting mechanisms.

\noindent \textbf{3) Proof Precomputation DoS}: The malicious users execute $\slap$ up to the location proof phase and precompute a batch of valid proofs within a limited time window. Since each proof embeds a timestamp and remains valid only for a bounded duration, the attacker must carefully balance the number of precomputed proofs to maximize server load while minimizing proof expiration. This scenario captures a more strategic adversary that exploits temporal validity to amplify attack impact. 

\subsubsection{Location Spoofing Attack Scenarios}
We consider two representative location spoofing scenarios derived from our location verification algorithms: (i) Distance Hijacking and (ii) Distance Fraud. They capture relay-based and protocol-level proximity violations and are evaluated in Section~\ref{subsec:simulation}. 
    
\noindent \textbf{1) Distance Hijacking}: A malicious UE that is physically far from the gNodeB attempts to impersonate a nearby device by relaying messages through a compromised benign $\ue$ within legitimate proximity. The benign UE is assumed to be honest but compromised, and thus serves as a relay for all protocol messages exchanged between the malicious UE and the gNodeB. By leveraging the strong received signal strength (RSS) of the compromised $\ue$, the attacker aims to deceive the gNodeB into inferring a shorter distance than the malicious UE’s true location. However, because the relayed communication introduces additional propagation delay, the measured round-trip time (RTT) reflects the true distance between the malicious UE and the gNodeB. This scenario allows us to study the relative effectiveness of RSS- versus RTT-based distance estimation and assess the system’s resistance to relay-based proximity attacks. The message flow of this attack is shown in Figure~\ref{fig:distance-highjacking}.
    
\noindent\textbf{2)  Distance Fraud}:  No access point or trusted infrastructure is available in the vicinity, and proximity verification relies on device-to-device DBPs with nearby nodes. A malicious UE that is physically distant attempts to authenticate itself as nearby by manipulating the DBP execution.     Since the DBP enforces proximity using both RTT constraints and challenge–response correctness, the adversary cannot reduce propagation delay below physical limits. Instead, the malicious UE attempts to guess the verifier’s challenges in advance and transmit corresponding responses prematurely, with the goal of appearing closer than it actually is. This scenario captures a protocol-level spoofing attack that does not rely on relaying, but instead exploits the probabilistic nature of challenge-response mechanisms under strict timing constraints.




\begin{figure}
    \centering
    \includegraphics[width=0.90\linewidth]{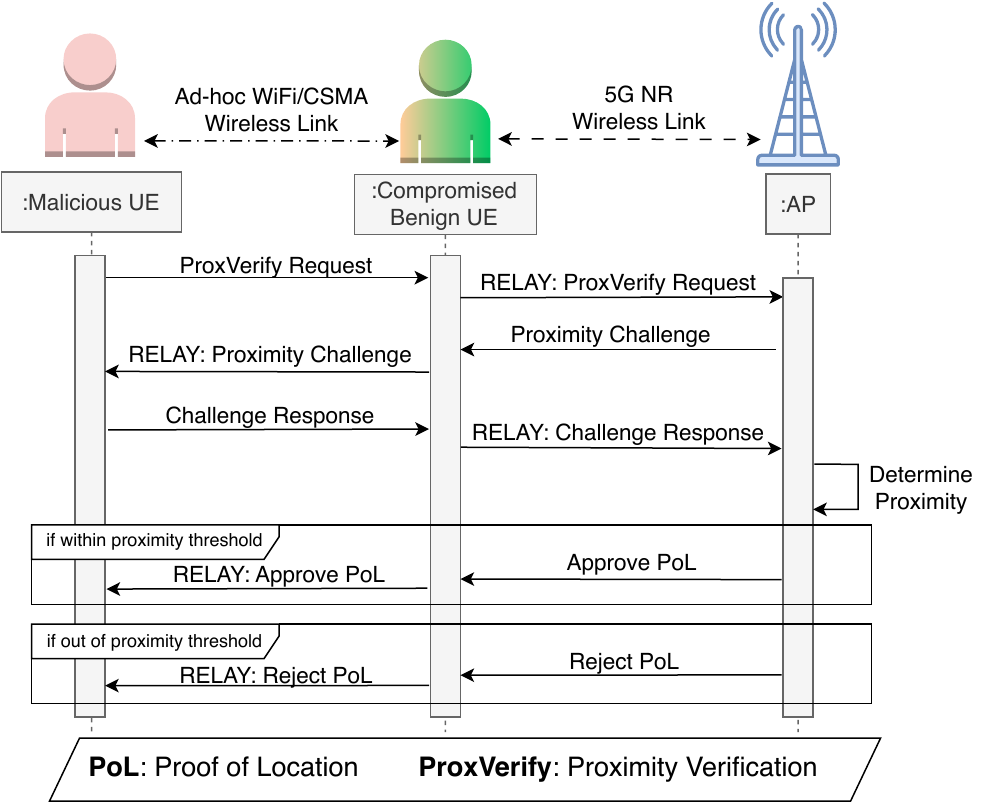}
    \caption{Distance hijacking flow using benign UE as a relay} 
    \vspace{-2mm}
    \label{fig:distance-highjacking}\vspace{-4mm}
\end{figure}





\section{$\slap$: Design and Instantiation}
\label{sec:scheme}
This section outlines the key enhancements of $\slap$ over the original SLAP protocol, and then presents the system setup, core operations, and detailed algorithmic specification of the $\slap$ framework. 

\noindent{\em \textbf{Improvements over SLAP Framework and the SOTA}:} $\slap$ extends our SLAP framework~\cite{darzi2025slap}, incorporating substantial new material and the following key enhancements and additional features:  
\emph{(i) Comprehensive DoS countermeasures:} We significantly strengthen DoS resilience by introducing on-demand, device-specific, and non-parallelizable VDFs, authenticated using lightweight standard signatures (e.g., ECDSA). This defense is further reinforced by a rate-limiting mechanism based on the linkability of event-oriented RLRSs, enabling effective DoS mitigation for both geolocation databases and service servers. 
\emph{(ii) Enhanced Proof of Location:} We enhance the location proof mechanism by leveraging RLRS-based event linkability to enforce rate limiting on proof issuance and enable revocation of APs that collude with malicious users. This provides accountability and supports more realistic adversarial network settings. 
\emph{(iii) Comprehensive performance evaluation and network simulations:} We substantially expand the evaluation by including detailed cryptographic benchmarking, comparisons with SOTA, and extensive network simulations that model real-world DoS and location spoofing scenarios under practical deployment conditions. 

$\slap$ advances the state of the art in DB-CRN security by providing the first unified, regulation-compliant framework that simultaneously achieves anonymous and location-private spectrum access, verifiable location assurance, and proactive DoS resilience under realistic adversarial conditions. Unlike prior solutions that address these challenges in isolation or rely on heavy cryptography, multiple non-colluding databases, or trusted infrastructure, $\slap$ integrates delegatable anonymous credentials, adaptive location proofs, VDF-enforced rate control, and RLRS-based accountability into a single lightweight architecture compatible with existing DB-CRN deployments. Compared to existing schemes, $\slap$ operates with a single geolocation database, incurs orders-of-magnitude lower communication overhead, and achieves substantially lower end-to-end latency, while uniquely preserving anonymity even against access points and databases. 

\subsection{$\slap$ Framework Initial Setup}
\label{subsec:initialsetup}
\noindent \textbf{DB Structure:} The DBs are organized in compliance with FCC spectrum-sharing regulations and maintain synchronized views of spectrum availability as required by the FCC. Each DB entry is indexed by location coordinates at a predetermined spatial resolution (e.g., 50 m), along with device characteristics (e.g., transmit power, device type) and access constraints (e.g., bandwidth/duration). This structured indexing ensures consistency across synchronized databases. Each record includes spectrum availability data, a validity time window, and operational parameters such as the maximum allowable transmission power (EIRP).

\noindent \textbf{Access Point Setup:} Access Points (APs) within a geographical region jointly form a ring ($\mathcal{L}_{\id}$) that contains the $\id$ of all participating APs. Each AP is provisioned with a cryptographic key pair $(\sk_\ap, \pk_\ap)$, generated by the FCC using the $\rlrs.\extract(pp, \id_\ap)$ algorithm (Definition~\ref{def:RLRS}). APs periodically broadcast time-synchronized beacons $\beta_\ts$ within predefined time windows $\ts$ to enable nearby device discovery.  
To verify UE proximity, each AP performs signal-based ranging using physical-layer measurements, including received signal strength ($\rss$) and round-trip time ($\rtt$)~\cite{abuyaghi2024positioning}. These measurements are processed via $\Delta \leftarrow \proxverify(\rss, \rtt, env_{\text{params}})$, which incorporates environment-specific factors such as path-loss models, multipath effects, and channel conditions to estimate the UE–AP distance and associated confidence bounds~\cite{lopez2019ieee}. The resulting estimate $\Delta$ determines whether the UE lies within the AP’s trusted coverage region, providing a confidence interval that underpins subsequent authenticated access decisions. Additional details are provided in~\cite{rosler2025improving}, with simulation results reported in Section~\ref{sec:performance}.


\noindent \textbf{System Setup:} In the system initialization, the FCC, acting as the core network authority, serves as the root issuer for all participants. For each registered UE, the FCC certifies a set of device attributes (e.g., identifier, transmit power) and access policies (e.g., validity period) by issuing a level-1 root credential via $\createcred(L', A, \sk_\fcc)$. A UE, represented by a pseudonym $\nym_c$, retrieves it through $\getcred(\pk_\fcc, \sk_c, A)$, obtaining $(\cred, (\overrightarrow{C}, \overrightarrow{O}), \dk_{L'})$. 
The issued credential $\cred$ binds the attribute commitment $\overrightarrow{C}$ to the FCC’s public key $\pk_\fcc$, includes the corresponding opening information $\overrightarrow{O}$, and embeds a delegation key $\dk_{L'}$ that permits controlled delegation up to level $L'$. This credential allows UEs to refresh their pseudonyms or delegate them to another party by transitioning to a new public key and, if needed, augmenting the attribute set to $A'$. Credential usage requires the UE to prove knowledge of the associated secret key $\sk_c$ and to generate a randomized signature over the relevant attributes.

\subsection{$\slap$ Framework Main Operations}
\label{subsec:mainoperations}
The high-level flow of the $\slap$ framework is depicted in Figure~\ref{fig:SLAP-benign-flow}. Moreover, a more detailed look at the internal steps of the protocol is provided in Algorithms~\ref{Alg:PoLAP}, \ref{Alg:PoLND}, and \ref{Alg:slap}. 

\begin{figure}
    \centering
    \includegraphics[width=\linewidth]{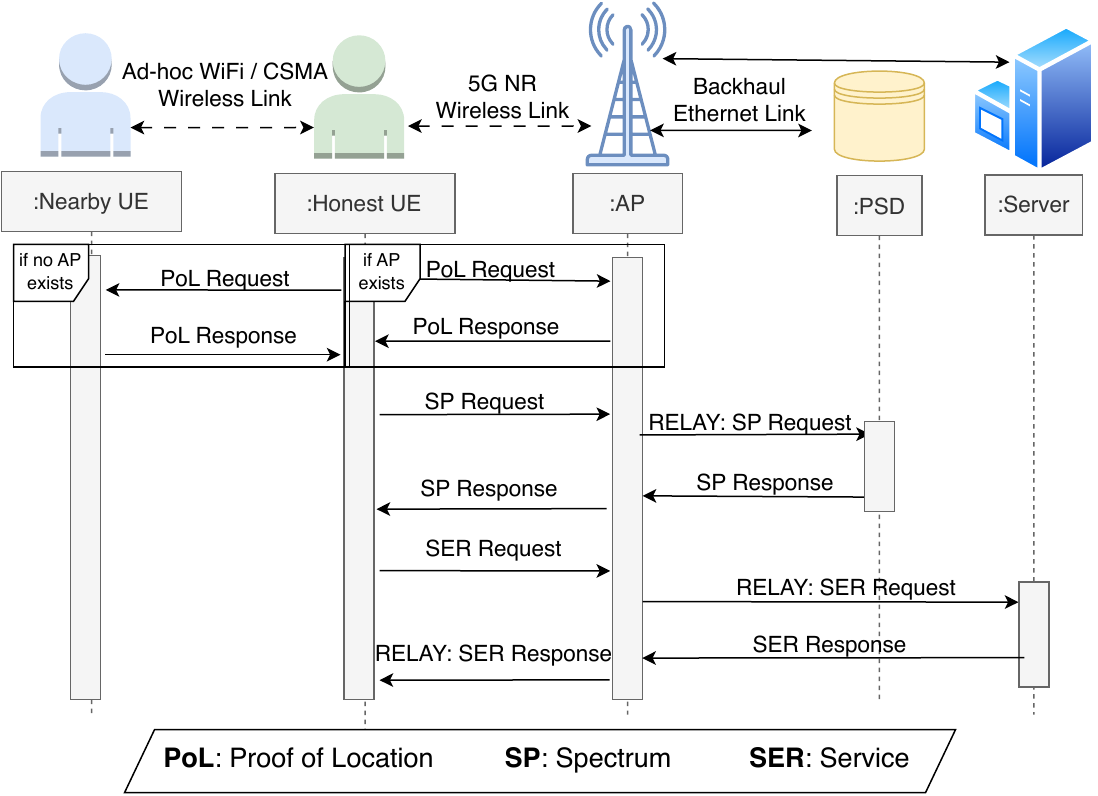}
    \caption{High-level flow of the $\slap$ protocol}
    \label{fig:SLAP-benign-flow}
    \vspace{-1em}
\end{figure}





\subsubsection{\textbf{Proof of Location Acquisition and Validation}}  Before querying PSDs or accessing network services, a UE must acquire a valid proof of location ($\pol$). This phase establishes a $\pol$ bound to a specific geographic region and time window. The process accommodates two complementary scenarios: infrastructure-rich settings with an available nearby AP, and resource-constrained or rural environments where no AP is present.

{\em (i) AP-Assisted $\pol$ Generation:} When a nearby AP is available, the client obtains a $\pol$ following Algorithm~\ref{Alg:PoLAP}. After receiving periodic beacons ($\beta_\ts$) from surrounding APs, the client selects the AP with the strongest signal and requests a $\pol$ for the current time window $\ts$. The client discloses the received beacon, its location coordinates, and the current timestamp as credential attributes and proves them anonymously to the AP. Upon validating the client’s credentials $(\nym_c, \cred_c)$ under the FCC’s public key, the AP assesses the client’s physical proximity using signal strength and RTT measurements via $\proxverify(\cdot)$. If the measured distance satisfies the predefined threshold, the AP binds the credential to the disclosed attributes, derives an event identifier, and generates an RLRS signature over the client’s attributes and credential (Steps~6–10). The signed tuple is returned to the client, who verifies the signature and accepts it as a valid proof of location $\Phi$ (Steps~11–12).

\begin{algorithm}[ht!]
	\small
	\caption{${\Phi \as \pol.\AP(\ts, (l_x,l_y), \cred_c, \beta_\ts)}$} \label{Alg:PoLAP}
	\hspace{5pt}
	\begin{algorithmic}[1]
        \vspace{-1mm}
        \Statex \hspace{-5mm}\squarenum{1} \underline{$D_\ts \as \client.\pol.\request(\beta_\ts)$}:\vspace{+1mm}
        \State Given the latest beacon $\beta_\ts$
        \State Set $D_\ts \as (\beta_\ts, (l_x, l_y), \ts)$
        \State Perform $\credprove(\pk_{\fcc}, \sk_{c}, \nym_c, \aux_c, \cred_c, D_\ts)$
        
        \algrule

        \vspace{-1mm}
        \Statex \hspace{-5mm}\squarenum{2} \underline{$(\textbf{m},\sigma_\ap, \tau_\ap) \as \ap.\pol.\respond(\pk_{\fcc}, \nym_c, \cred_c, D_\ts)$}:\vspace{+1mm}
        \State \textbf{if}~$1 \as \ap.\credverify(\pk_{\fcc}, \nym_c, D_\ts)$, \textbf{then}
        \State ~~~~$\Delta_c \as \ap.\proxverify(\rss, \rtt, env_{param})$
        \State ~~~~\textbf{if}~$(l_x, l_y) \in \Delta_c$, \textbf{then}
        \State ~~~~~~~~Set $\textbf{m}\as \{D_\ts, \nym_c, \cred_c\}$ 
        \State ~~~~~~~~Set $e_\id\as\{(l_x,l_y), \ts, \beta_\ts\}$
        \State ~~~~~~~~$(\sigma_\ap, \tau_\ap) \as \rlrs.\sign(\sk_\ap, \textbf{m}, \mathcal{L}_\id, e_\id)$ 
        \State ~~~~~~~~Send ($\textbf{m},\sigma_\ap, \tau_\ap$) to the Client.
        
	\algrule
    
    \vspace{-1mm}
        \Statex \hspace{-5mm}\squarenum{3} \underline{$\Phi \as \client.\pol.\verify(\textbf{m},\sigma_\ap, \tau_\ap)$}:\vspace{+1mm}
        \State \textbf{if}~$\{0,1\} \as \rlrs.\verify(\mathcal{L}_\id, \textbf{m}, e_\id, \sigma_\ap, \tau_\ap)$, \textbf{then}
        \State ~~~~\Return~$\Phi \as (\textbf{m}, \sigma_\ap, \tau_\ap)$
        \end{algorithmic}
\end{algorithm} 
\setlength{\textfloatsep}{0pt}

{\em (ii) Device-Assisted $\pol$ Generation:} In environments with limited infrastructure, such as rural areas without nearby WiFi APs or cellular coverage, the client relies on neighboring devices (NDs) to obtain a $\pol$ and anonymous credentials, as described in Algorithm~\ref{Alg:PoLND}. The client initiates this process by broadcasting a $\pol$ request that includes the current time and claimed location (Step~1). Upon receiving replies, the client designates its location coordinates and timestamp as the attributes to be selectively disclosed and proves them anonymously to the responding ND. Using the FCC’s public key, the ND verifies the client’s credential and, upon successful validation, establishes a shared session key $ss$ through an interactive authenticated key agreement. The ND then executes a symmetric-DBP to confirm the client’s proximity within a predefined threshold $\textit{th}$ (Steps~4–6). If the proximity check succeeds, the ND augments its attribute set with the client’s location $(l_x, l_y)$ and time window $\ts$, and issues a delegated credential with restricted delegation rights (Steps~7–10). Finally, the client retrieves the delegated credential and corresponding location proof, authenticated under the FCC’s public key. Thereby, obtaining a valid $\pol$ embedded within the extended attributes (Steps~11–13).

\begin{algorithm}[ht!]
	\small
	\caption{${(\cred'_c, A') \as \pol.\ND(\ts, (l_x,l_y), \cred_c)}$}\label{Alg:PoLND}
	\hspace{5pt}
	\begin{algorithmic}[1]
        \vspace{-1mm}
        \Statex \hspace{-5mm}\squarenum{1} \underline{$D_\ts \as \client.\pol.\request(\nym_c, \cred_c)$}:\vspace{+1mm}   
        \State Request $\pol$ and a delegated credential $\cred'_c$ from an $\nd$
        \State Set $D_\ts \as ((l_x, l_y), \ts)$
        \State Perform $\credprove(\pk_{\fcc}, \sk_{c}, \nym_c, \aux_c, \cred_c, D_\ts)$
    \algrule
	\vspace{-1mm}
        \Statex \hspace{-5mm}\squarenum{2} \underline{$D_\ts \as \nd.\pol.\respond(\beta_\ts)$}:\vspace{+1mm}
        \State \textbf{if}~$1 \as \nd.\credverify(\pk_{\fcc}, \nym_c, D_\ts)$, \textbf{then}
        \State ~~~~Perform $ss \as \aka(\sk_\nd, \pk_\nd, \pk_c)$
        \State ~~~~\textbf{if}~$1\as\sym\dbp(ss, \textit{th})$ and $(l_x,l_y) \in \textit{th}$, \textbf{then}
        \State ~~~~~~~~Set $A_l \as D = ((l_x, l_y), \ts)$
        \State ~~~~~~~~Set $dk'_{L''} := \perp$
        \State ~~~~~~~~$\issuecred(\pk_{\fcc}, \dk_{L'}, \sk_\nd, \cred_\nd, A_l, L'')$ 
        \State Send $(\cred'_c, \dk'_{L''})$ to the client.

	\algrule
    
	    \vspace{-1mm}
        \Statex \hspace{-5mm}\squarenum{3} \underline{$D_\ts \as \client.\pol.\verify(\beta_\ts)$}:\vspace{+1mm}
            \State $(\cred'_c, \dk'_{L''}) \as \receivecred(\pk_{\fcc}, \sk_c, A_l)$
            \State Set $A' \as (A, \Phi)$ where $\Phi \as A_l$
            \State \Return ($\cred'_c, A'$)
        \end{algorithmic}
\end{algorithm} 
\setlength{\textfloatsep}{0pt}

\subsubsection{\textbf{Querying Spectrum Availability and Services}} 
Algorithm~\ref{Alg:slap} describes the procedure for requesting spectrum availability and related services. Each query must be preceded by obtaining a valid $\pol$ bound to the client’s coordinates and the current timestamp.   
In the \emph{AP-assisted} $\pol$ case, the client initiates the request by anonymously proving its credential while disclosing its location, time window $\ts$, and the obtained location proof $\Phi$, and then submits a spectrum query for the surrounding area (Steps~1–5). The PSD first validates the client’s credential and then verifies the RLRS signature issued by the AP. Using the RLRS-based $\pol$ within the current time interval, the AP checks the linkability of the proof against other signatures to detect misuse such as DoS and replay (Steps~12–14). Upon successful verification and based on the disclosed device attributes (e.g., transmit power and device type), the PSD generates a VDF with an appropriate difficulty level and signs it using $\sgn$ (Steps 10-14).  
In the \emph{device-assisted} $\pol$ case, the client relies on a delegated credential that embeds the location, timestamp $\ts$, and $\Phi$ as disclosed attributes. The client proves this credential and submits a spectrum query in the same manner. After validating the credential, the PSD follows an identical process by issuing a VDF with the corresponding difficulty and signing it via $\sgn$, then returns the spectrum availability information together with the VDF and its signature to the client.

\begin{algorithm}[ht!]
	\small
	\caption{$\slap$ Scheme}\label{Alg:slap}
	\hspace{5pt}
	\begin{algorithmic}[1]

		\vspace{-1mm}
        \Statex \hspace{-5mm}\squarenum{1} \underline{$\rho_c \as \client.\query(\cred_c, \Phi_c, (l_x,l_y), \ts, freq)$}:\vspace{+1mm} 
            \State Given $(l_x, l_y)$ and $\ts$, client request $\pol$: 
            \State \textbf{If}~${\Phi_c \as \client.\pol.\AP((l_x,l_y), \ts, \cred_c, \beta_\ts)}$
            \State ~~~~Set $D_c \as ((l_x, l_y), \tv, ch, \Phi)$
            \State ~~~~Perform $\credprove(\pk_{\fcc}, \sk_{c}, \nym_c, \aux_c, \cred_c, D_c)$
            \State ~~~~Query a $\psd$ for $\rho_c \as ((l_x,l_y), ch, \tv)$     
            \State \textbf{elseif}~${(\cred'_c, A') \as \client.\pol.\ND(\cred_c, (l_x, l_y), \ts)}$
            \State ~~~~Given $\cred_c \as \cred'_c$ with $A \as (A, \Phi)$ for $((l_x,l_y), \tv)$
            \State ~~~ Perform $\credprove(\pk_{\fcc}, \sk_{c}, \nym_c, \aux_c, \cred_c, A')$
            \State ~~~ Query a $\psd$ for $\rho_c \as ((l_x,l_y), ch, \tv)$
    \algrule
		\vspace{-1mm}
        \Statex \hspace{-5mm}\squarenum{2} \underline{$R_\psd \as \psd.\respond((l_x,l_y), ch, \tv)$}:\vspace{+1mm}
            \State \textbf{If}~$1 \as \credverify(\pk_{\fcc}, \nym_c, D)$
            \State ~~~~\textbf{If}~${1 \as \rlrs.\verify(\vk, \textbf{m}=(D,\nym_c, \cred_c), \sigma_{\ap})}$
            \State ~~~~~~~~\textbf{for all}~$(m', \sigma', \tau') \in \ts$~\textbf{do}
            \State ~~~~~~~~~~~~\textbf{If}~${1\as\rlrs.\link(\mathcal{L}_\id, e_\id, m, \sigma, \tau, m', \sigma', \tau')}$
            \State ~~~~~~~~~~~~~~~~\textbf{Return}~$\perp$
            \State ~~~~~~~~~~~~\textbf{else}
            \State ~~~~~~~~~~~~~~~~Set $\Pi \as \vdf.\setup(1^\lambda, \kappa)$ accordingly
            \State ~~~~~~~~~~~~~~~~$\sigma_{\pi_\theta}\as \sgn.\sign(\sk_\psd, \pi_\theta)$
            \State ~~~~~~~~~~~~~~~~\textbf{Return}~$R_\psd \as (\Pi, \sigma_{\pi_\theta}, \textit{freq})$
            \State \textbf{elseif} $1 \as \credverify(\pk_\fcc, \nym_c, A')$
            \State ~~~~Set $\Pi \as \vdf.\setup(1^\lambda, \kappa)$ accordingly
            \State ~~~~~~~~~~~~$\sigma_{\pi_\theta}\as \sgn.\sign(\sk_\psd, \pi_\theta)$
            \State ~~~~\textbf{Return}~$R_\psd \as (\Pi, \sigma_{\pi_\theta}, \textit{freq})$
    \algrule
        \vspace{-1mm}
        \Statex \hspace{-5mm}\squarenum{3} \underline{$\{0,1\}\as\client.\service.\request(\Phi, (l_x,l_y), \ts)$}:\vspace{+1mm} 
        \Statex \vspace{-1mm} Notifying spectrum usage or requesting services to $\server$s:
	    \State Given the request/notification message as $m$
        \State $(\ell, \psi) \as \vdf.\eval(m, \tau)$
        \State Performs $\credprove(\pk_{\fcc}, \sk_{c}, \nym_c, \aux_c, \cred_c, D_c)$
        \State Send $(m, \psi)$ to $\psd$/$\server$
    \algrule
		\vspace{-1mm}
        \Statex \hspace{-5mm}\squarenum{4} \underline{$\{0,1\} \as\server.\respond(m, \psi, )$}:\vspace{+1mm} 
            \State \textbf{If}~$1 \as \credverify(\pk_{\fcc}, \nym_c, D)$
            \State ~~~~\textbf{If}~$1 \as \sgn.\verify(\pi_\theta, \sigma_{\pi_\theta})$
            \State ~~~~~~~~\textbf{If}~${1 \as \rlrs.\verify(\vk, \textbf{m}=(D,\nym_c, \cred_c), \sigma_{\ap})}$
            \State ~~~~~~~~~~~~$\PSD$/$\server$ \textbf{Return} $1$, and Grant Services.
            \State \textbf{elseif} $1 \as \credverify(\pk_\fcc, \nym_c, A')$
            \State ~~~~\textbf{If}~$1 \as \sgn.\verify(\pi_\theta, \sigma_{\pi_\theta})$
            \State $\PSD$/$\server$ \textbf{Return} $1$, and Grant Services.
        \end{algorithmic}
\end{algorithm} 
\setlength{\textfloatsep}{0pt}

\subsubsection{\textbf{Notifying Spectrum Usage and Service Requests}}  
To submit spectrum usage reports or request network services, the client first evaluates the previously issued VDF through sequential computation and derives a communication token bound to the message $m$, which encodes either usage information or a service request (Steps 19–20). The client then presents this token and anonymously proves possession of a valid credential to the server (or equivalently to the PSD) to initiate the request. Upon successful verification of the credential and validation of the VDF, the server (or PSD) checks the associated proof of location according to whether it was generated via the AP-assisted or device-assisted $\pol$ process. If all checks pass, the PSD records the reported spectrum usage, or the server authorizes access to the requested services. Distinct from prior approaches, spectrum usage reporting in $\slap$ is also protected by anonymous credentials with attribute binding, enabling higher-quality usage information while remaining compliant with FCC coexistence requirements.

\vspace{+3mm}
\subsection{Security Analysis}
\label{subsec:security}
Based on the established threat model, we formalize the security model and present the corresponding security analysis. We state the main security theorem for $\slap$, while the complete security model, formal definitions, and detailed proofs are provided in Appendix~\ref{sec:appendix_security}.

\begin{theorem} \label{the:slpaTheorem1}
The $\slap$ framework achieves the following security guarantees:  
(i) anonymous user authentication, ensured by the anonymity, soundness, and unforgeability of the underlying ZKPoK and SPSEQ-UC signature schemes;  
(ii) location privacy, provided by the unlinkability of credential-based signatures;  
(iii) verifiable user location, enforced through either the unforgeability of RLRS combined with signal-based proximity measurements or the security of public-key distance-bounding protocols and anonymous credential delegation; and  
(iv) resilience to DoS attacks, enabled by VDF and rate-limiting mechanisms.
\end{theorem}

Additionally, Section~\ref{subsec:simulation} presents a network-level security analysis of $\slap$, including simulation-based evaluations of its resilience to various DoS and location spoofing attack scenarios in wireless network settings.

\section{Performance Evaluation}
\label{sec:performance}
This section provides a comprehensive evaluation of the $\slap$ framework, covering both cryptographic primitive-level and network-level performance analysis.

\subsection{Primitive-Level Evaluation} \label{subsec:SalehEval} We evaluate the primitive-level performance of $\slap$ by first describing the evaluation metrics and implementation setup, followed by an empirical analysis of the framework and a comparison with state-of-the-art schemes.

\subsubsection{Configuration and Experimental Setup}
\label{subsec:Configuration}

\noindent \textbf{Hardware:} We evaluate the efficiency of $\slap$ on a standard desktop platform equipped with an Intel Core $\textit{i9-11900K}$ $@ 3.50~GHz$, $64~\text{GiB}$ RAM, $1~\text{TB}$ SSD running Ubuntu $22.04.4~\text{LTS}$. 

\noindent \textbf{Libraries:} The implementation spans multiple languages (including C, Python, and Java) and leverages established cryptographic libraries and tools such as \textit{DAC-from-EQS}\footnote{\url{https://github.com/mir-omid/DAC-from-EQS}}, \textit{Wesolowski-VDF}\footnote{\url{https://github.com/futexor/Wesolowski-VDF}}, \textit{RELIC}\footnote{\url{https://github.com/relic-toolkit/relic}}, and \textit{OpenSSL}\footnote{\url{https://www.openssl.org/}}, to realize the core components of $\slap$ and its underlying primitives (e.g., hashing, modular arithmetic, and exponentiation).  

\noindent \textbf{Configuration and Parameter Selection:} We exclude signal transmission and network communication delays from our analysis, as these are typically constant-time operations and occur at the microsecond scale. Our configuration employs SHA-256 for hashing, structure-preserving set commitments~\cite{fuchsbauer2019structure}, SPSEQ-UC signatures~\cite{mir2023practical}, and the BN-256 elliptic curve (EC) for binding and EC operations (with ${n=2048}$). Anonymous authentication is realized via Schnorr-type zero-knowledge proofs augmented with Damgård’s technique, following the DAC construction~\cite{mir2023practical}. Non-interactive proofs are instantiated via the Fiat–Shamir transform, yielding an effective security level of approximately $100$ bits. The $\sgn$ digital signature is instantiated using ECDSA, a widely adopted and computationally efficient standard signature scheme.


\noindent \textbf{Evaluation Metrics and Rationale:} We empirically analyze the $\slap$ framework by measuring cryptographic, computational, and communication overhead across all operational phases and underlying primitives, including DAC, RLRS, DBP, and VDF. Since no existing solution provides an equivalent combination of functionalities, direct performance comparisons are not applicable. Instead, we present a comprehensive performance breakdown of $\slap$ along key metrics to assess its practicality and deployability. We further include a qualitative and analytical comparison with representative schemes that address partial aspects of these functionalities in the DB-CRN setting. The evaluation proceeds as outlined below.

\subsubsection{Experimental Results} \label{subsec:Performance}
The empirical evaluation of cryptographic overhead, computational costs, and communication overhead for each phase of the $\slap$ framework is summarized in Table~\ref{tab:Computational} and detailed below:

\begin{table}[]
\caption{Computational and communication overhead of $\slap$}
    \centering
    \renewcommand{\arraystretch}{1.2} 
    \setlength{\tabcolsep}{2pt}
    \begin{tabular}{|@{}c@{}|@{}c@{}|@{}c@{}|@{}c@{}|}\hline
    \multirow{2}{*}{\textbf{Phase}} & \multirow{2}{*}{\textbf{Entity}} & \textbf{Computational} & \textbf{Communication} \\
    & & \textbf{Cost (\textit{ms})} & \textbf{Overhead (\textit{B})} \\\hline
    \multirow{2}{*}{\textbf{PoL.AP}}& \textit{Client} & $32.2$ & \multirow{2}{*}{$2456$}\\\cline{2-3} 
    & \textit{AP} & $75.79$ & \\\hline\hline
    \multirow{3}{*}{\textbf{PoL.ND}} & \textit{Client} &  $31.8$ & \multirow{2}{*}{$1944$}\\ \cline{2-3}
    & \textit{ND} & $77.05$ & \\ \hline\hline
    \textbf{Spectrum} & \textit{Client} & $17.23$ & \multirow{2}{*}{$3016$} \\\cline{2-3} 
    \textbf{Query}& \textit{PSD} & $74.34$ & \\ \hline
    \textbf{Service}& \textit{Client} & $\kappa\times S_q + 17.23$ & \multirow{2}{*}{$2712$} \\
    \textbf{Request} & \textit{Server} & $77.68$ & \\\hline    
    \end{tabular}
    \begin{tablenotes}
       \item \textbf{Note:} $S_q$ and $\kappa$ represent the repeated squaring time to solve a puzzle and puzzle difficulty, respectively. 
    \end{tablenotes}\vspace{+2mm}
    \label{tab:Computational}
\end{table}

\noindent\textbf{Cryptographic Overhead:}
To demonstrate credential ownership, the user re-randomizes both the credential $\cred$ and pseudonym $\nym$, and executes a ZKPoK over the secret key $\sk$ and auxiliary randomness $\aux$. This process yields a fresh, unlinkable pseudonym and selectively reveals an attribute subset $D$ via a set commitment construction.  
In practice, the dominant runtime overhead arises from signature conversion, representation normalization, and re-binding to a new set commitment, which require approximately $2$ ms, $5$ ms, and $13$ ms, respectively. 
In the adopted PK-DBP, the interactive AKA phase incurs one ECC multiplication ($0.61~\textit{ms}$), a single hash evaluation ($0.35~\textit{ms}$), and random value generation ($0.05~\textit{ms}$). The subsequent rapid bit-exchange phase operates at nanosecond granularity; the resulting distance-fraud error margin (on the order of ${1-10~\textit{m}}$) is negligible relative to the remaining protocol costs. Proximity verification at the AP, implemented via RSS and RTT measurements in $\proxverify(\cdot)$, typically completes within $1–10~\textit{ms}$ depending on client distance and environmental conditions. 
When instantiated using ECDSA, the $\sgn$ operation requires about $0.3~\textit{ms}$ for signature generation and $3.8~\textit{ms}$ for verification. For the RLRS component, signing, verification, linking, and revocation require approximately $15.8~\textit{ms}$l, $14.98~\textit{ms}$, $0.12~\textit{ms}$, and $0.6~\textit{ms}$, respectively. These costs can be further reduced at the PSD through batch verification and parallelized linking. Finally, solving a VDF on commodity hardware incurs increasing latency with higher difficulty parameters $\kappa$: about $17~\textit{ms}$ for $10^3$ squarings, $61~\textit{ms}$ for $5\times10^3$, $121~\textit{ms}$ for $10^4$, $874~\textit{ms}$ for $8\times10^4$, and $3.17~\textit{s}$ for $3\times10^5$. In contrast, VDF setup and verification remain constant-time, averaging roughly $89~\textit{ms}$ and $4~\textit{ms}$, respectively, across all difficulty levels.

\noindent\textbf{Computational Costs:}
{\em (i)} \textit{$\pol.\ap$ Phase:} The client proves its credential $\cred_c$ and verifies the RLRS signature, while the AP validates the credential, executes $\proxverify$, and issues an RLRS signature as the location proof. {\em (ii)} \textit{$\pol.\nd$ Phase:} Two users engage in interactive protocols that include credential proving and verification, key agreement AKA, symmetric-DBP, and credential delegation with the location proof embedded as a certified attribute. 
{\em (iii)} \textit{$\query$ Phase}: The client submits a valid $\pol$ (from either an AP or an ND) to query the PSD for spectrum availability at a given location and time. The PSD verifies the location proof and linkability (via RLRS or $\cred_c$), then returns spectrum data together with a user-specific VDF. {\em (iv)} \textit{$\service.\request$ Phase:} The client evaluates the VDF, proves possession of a valid credential, and submits the result. The server (or PSD) verifies the solution before authorizing service access or updating the $\db$.

\noindent\textbf{Communication Overhead:}
We consider groups of size ${|\mathbb{G}_1| = |\mathbb{Z}_p| = 256}$ bits, ${|\mathbb{G}_2| = 512}$ bits, and ${|\mathbb{G}_T| = 3072}$ bits, with arithmetic performed modulo a 2048-bit integer. Message payloads satisfy ${|m| < 256}$ B, timestamps occupy $8$ B on a 64-bit Unix system, precise location coordinates require $16$ B, and spectrum availability data $\beta$, derived from FCC records, is approximately $560$ B. All attributes in the system are assumed to have uniform size.  
Credentials encapsulate the values $|\cred|$, $|\sk|$, and $|\nym|$ within set commitments and SPSEQ-UC signatures, yielding a constant-size representation that does not depend on the number of attributes. Specifically, the credential size is ${4|\mathbb{G}_1| + |\mathbb{G}_2| + |\mathbb{Z}_p|}$, totaling $1792$ bits. The size of the commitment vector $\overrightarrow{C}$ scales with the delegation depth (here ${L=2}$), resulting in communication overhead that grows linearly with both the number of attributes and delegation levels. The RLRS signature is independent of the number of APs and contains ${\sigma = 2|G_1| + 3G_1 + 15Z_p =640}$~bytes. Using publicly available FCC database records\footnote{\url{https://enterpriseefiling.fcc.gov/dataentry/public/tv/lmsDatabase.html}}, we estimate each spectrum database entry to be roughly $560$ bytes, augmented with synthetic data for evaluation. The aggregate communication costs across all protocol phases are summarized in Table~\ref{tab:Computational}.


\begin{table*}[ht]
\caption{Qualitative and analytical comparison with existing location privacy schemes}
    \centering
    \resizebox{0.99\textwidth}{!}{
    \Large
    \renewcommand{\arraystretch}{1.2} 
    \begin{tabular}{|{l}|{c}|{c}|{c}|{c}|{c}||{c}|{c}|{c}|{c}|{c}|}\hline
         \multirow{2}{*}{\textbf{Scheme}} & \multicolumn{5}{c||}{\textbf{Features}} & \multicolumn{4}{c|}{\textbf{Delay}} & \textbf{Total} \\ \cline{2-10}
         & \textbf{Setting} & \textbf{Loc.Privacy} & \textbf{Anonym} & \textbf{Loc.Verification} & \textbf{Counter-DoS} & \textbf{SU} & \textbf{PSD} & \textbf{E2E} & \textbf{PoL} & \textbf{Communication}  \\\hline\hline
         \textit{Troja et al}~\cite{troja2014leveraging} & \textit{1-DB} &  \textit{Peer-to-Peer} & \boldsymbol{\cmark} & \boldsymbol{\xmark}& \boldsymbol{\xmark} & $1650$ \textit{ms} & $11760$ \textit{ms} & $13410$ \textit{ms} & \boldsymbol{\xmark} & $12$ \textit{MB} \\\hline
         \textit{Li et al}~\cite{li2015privacy} & \textit{1-DB} & \textit{Pseudo-ID} & \boldsymbol{\xmark} & \textit{WiFi AP+Loc.Server} & \boldsymbol{\xmark} & \boldsymbol{\xmark} & \boldsymbol{\xmark} & \boldsymbol{\xmark} & $210$ \textit{ms} & \boldsymbol{\xmark} \\\hline
         \textit{Xin et al}~\cite{xin2016privacy} &\textit{1-DB} & \textit{PIR} & \boldsymbol{\xmark} & \textit{WiFi AP+QRA} &\boldsymbol{\xmark} & $292.8$ \textit{ms} & $142.7$ \textit{ms} & $407.4$ \textit{ms} & $430.1$ \textit{ms} & 325\textit{KB}\\\hline
         \textit{LP-Chor}~\cite{grissa2019location} & \textit{$\ell$-DB} & \textit{PIR} &\boldsymbol{\xmark} & \boldsymbol{\xmark} & \boldsymbol{\xmark} & $7.7$ \textit{ms} & $480$ \textit{ms} & $620$ \textit{ms} & \boldsymbol{\xmark} & $753$ \textit{KB}\\\hline
         \textit{LP-Goldberg}~\cite{grissa2019location} & \textit{$\ell$-DB} & \textit{PIR} &\boldsymbol{\xmark} & \boldsymbol{\xmark} & \boldsymbol{\xmark} & $320$ \textit{ms} & $1210$ \textit{ms} & $1780$ \textit{ms} & \boldsymbol{\xmark} & $6$ \textit{MB}\\\hline
         \textit{RAID-LP-Chor}~\cite{grissa2019location} & \textit{$\ell$-DB} & \textit{PIR} &\boldsymbol{\xmark} & \boldsymbol{\xmark} & \boldsymbol{\xmark} & 
         $0.4$ \textit{ms} & $22$ \textit{ms} & $210$ \textit{ms} & \boldsymbol{\xmark}& $125$ \textit{KB}\\\hline
         Zeng et al~\cite{zeng2019efficient} & \textit{1-DB} & \textit{BS+ECC} & \textit{PseudoID} &\boldsymbol{\xmark} &\boldsymbol{\xmark}  & $87$ \textit{ms} & $27$ \textit{ms} & $135$ \textit{ms} & \boldsymbol{\xmark} & $1.24$ \textit{KB}\\ \hline
         \textit{TrustSAS}~\cite{grissa2021anonymous} &\textit{$\ell$-DB} &\textit{PIR} & \textit{EPID} & \boldsymbol{\xmark} & \boldsymbol{\xmark} & 
         $329.4$ \textit{ms} & $324.6$ \textit{ms} & $4954$ \textit{ms} & \boldsymbol{\xmark} & $1.25$ \textit{MB} \\\hline
         \textit{PACDoSQ}~\cite{darzi2024privacy} & \textit{$\ell$-DB} &\textit{PIR}& \textit{Tor} & \boldsymbol{\xmark} & \textit{HBP} & $28.1$ \textit{ms} & $199$ \textit{ms} & $1373.6$ \textit{ms} & \boldsymbol{\xmark} & $605.92$ \textit{KB}\\\hline
         \multirow{2}{*}{$\slap$} & \multirow{2}{*}{\textit{1-DB}} & \multirow{2}{*}{\textit{DAC}} & \multirow{2}{*}{\textit{DAC}} & \textit{WiFi-AP+RLRS} & \textit{VDF} &\multirow{2}{*}{$17.23$ \textit{ms}}&\multirow{2}{*}{$74.34$ \textit{ms}}&\multirow{2}{*}{$78.61$ \textit{ms}}&$107.99$ \textit{ms}& \multirow{2}{*}{$3.72$\textit{KB}}\\ 
         &&&&\textit{DBP+DAC}& \textit{Rate-Limiting} &&&&$108.85$ \textit{ms} &\\\hline
    \end{tabular}}
    \begin{tablenotes}
       \item \textbf{Libraries:} Virtual Machines running Ubuntu simulated PIR costs, using the \textit{percy++} library\footnote{\url{https://percy.sourceforge.net/}} for multi-server PIR, the \textit{Open Quantum-Safe} library\footnote{\url{https://openquantumsafe.org/}} for PQC primitives, and \textit{OpenSSL} for cryptographic operations and arithmetic. \textbf{Variables:} We consider six databases for multi-DB schemes with $|DB| = 560~\textit{MB}$ and $400$ rows/columns as described in \cite{xin2016privacy}. Key terms include \textit{BS} (base station), \textit{HBP} (hash-based puzzles), \textit{QRA} (quadratic residue assumption), and \textit{EPID} (enhanced privacy ID based on direct anonymous attestation). 
    \end{tablenotes}\vspace{-5mm}
    \label{tab:QuantitativeComparison}
\end{table*}

\noindent\textbf{Comparison with SOTA:}
We present a qualitative and analytical comparison of the functionality provided by $\slap$ against representative location-privacy schemes, summarized in Table~\ref{tab:QuantitativeComparison}. To ensure a balanced assessment, we evaluate spectrum query overhead from both the client and PSD perspectives, overall communication costs, and the end-to-end latency required to retrieve a single geolocation database record as an indicator of system scalability. As reflected in Table~\ref{tab:QuantitativeComparison}, $\slap$ uniquely supports the full set of security and privacy requirements for anonymous spectrum access, while enabling flexible deployment options for location verification and achieving low latency with minimal communication overhead.



\subsection{Network-Level Evaluation}
\label{subsec:simulation}
This section presents an instantiation of the SLAPX protocol in a realistic wireless network setting and evaluates its performance through simulations. We use \texttt{ns3}~\cite{NS3} to model the deployment of SLAPX and analyze its communication overhead, while the 5G NR (New Radio) LENA module~\cite{5G_LENA} is employed to simulate the 5G infrastructure.

\subsubsection{Simulation Setup and Configuration} \label{subsec:simulationsetup} 
\textbf{Overview of the \texttt{ns3} Setup:}
We assume that the regulatory authority (e.g., the FCC) has already provisioned all required certificates, keys, and credentials. As a result, the FCC is not explicitly modeled in the \texttt{ns3} environment. The remaining entities in the SLAPX framework are mapped as follows: the client is modeled as a 5G NR User Equipment (UE), the AP as a 5G gNodeB, and the PSD as a server node. Furthermore, servers that provide cloud services are modeled as one single server node. In this setup, each UE operates within the radio access network (RAN) and connects to the nearest gNodeB over a wireless channel using 5G NR links. Server nodes are deployed in the core network and are connected to the gNodeBs via wired backhaul links.


\textbf{5G Parameter Configuration:}
The 5G NR LENA module offers a wide range of configuration options for realistic 5G simulations. For simplicity, we retain default settings for most parameters and manually configure only a selected subset. The configured parameters and their values are summarized in Table~\ref{tab:ns3-configuration}. Specifically, the channel model follows the standard NR specification~\cite{3GPPChannelModel} used by default in the LENA module. The transmission power of both UE and gNodeB nodes is set to $30$~dBm, reflecting typical deployment conditions. The Maximum Transmission Unit (MTU) is set to the standard Ethernet frame size of $1500$~bytes, and the backhaul network is modeled with a data rate of $1$~Gbps and an end-to-end latency of $2$~ms.


\begin{table}[]
\caption{Simulation parameters in the \texttt{ns3} experiments}
    \centering
    \begin{tabular}{|p{3cm}|p{4cm}|} \hline
         \textbf{Parameter} & \textbf{Value} \\ \hline
         Channel Model & 3GPP TR 38.901 (default) \\ \hline
         UE Tx Power & 30 dBm \\ \hline
         gNodeB Tx Power & 30 dBm \\ \hline
         Number of cells per gNodeB & 3 \\ \hline
         Beam width & 50 degrees \\ \hline
         Backhaul data rate & 1 Gbps \\ \hline
         Backhaul delay & 2 ms \\ \hline
         MTU (DoS) & 1500 bytes \\ \hline
         MTU (Fragmentation) & 576-9000 bytes \\ \hline
    \end{tabular}
    \label{tab:ns3-configuration}
\end{table}

\textbf{Server Modeling:}
Because \texttt{ns3} is a discrete-event simulator, directly modeling real-time server overload caused by DoS attacks is nontrivial. While network-level congestion can be simulated, capturing the impact of DoS attacks on server-side processing requires an abstracted model. Therefore, each server is represented using a queue-and-worker abstraction, a common approach in analytical DoS modeling~\cite{Wang2007}. 
Specifically, the server consists of two components: \emph{(i) Workers}, which model independent computational units (e.g., CPU cores) capable of processing requests in parallel; and \emph{(ii) Queue}, which buffers incoming requests when all workers are busy. However, when the queue is full, incoming requests are dropped. This model enables us to capture server saturation and request loss under DoS conditions, serving as the basis for evaluating DoS scenarios in our simulations.

\textbf{Terminology:}
To present the evaluation results clearly and concisely, we use a set of symbols to represent the performance metrics and experiment parameters used in our DoS scenarios, as shown in Table~\ref{tab:DoS-symbols}.

\begin{table}[]
\caption{Symbols used in the scope of DoS experiments}
    \centering
    \begin{tabular}{|c|p{6cm}|} \hline
         \textbf{Symbol} & \textbf{Description} \\ \hline
         $N_U$ & Total number of UEs spawned \\ \hline
         $N_Q$ & Total number of messages that are queued \\ \hline
         $N_D$ & Total number of dropped benign messages due to busy workers and full queue capacity \\ \hline
         $t_Q$ & The average waiting time of a message in the queue (milliseconds) \\ \hline
         $r_{mal}$ & The ratio of malicious UEs to the total number of UEs \\ \hline
    \end{tabular}
    \label{tab:DoS-symbols}
\end{table}


\subsubsection{Simulation Results}  \label{subsec:simulation_results}

\textbf{Fragmentation Analysis of SLAPX:} Based on a realistic 5G deployment, we evaluate SLAPX through numerical analysis focusing on fragmentation behavior and communication overhead. Due to the relatively large size of certain protocol messages, packet fragmentation becomes observable as illustrated in Figure~\ref{fig:Fragmentation-Analysis}, where the upper illustration shows the relationship between MTU values and fragmentation across different SLAPX message types. Although the 5G standard does not mandate a fixed MTU, many deployments adopt the standard Ethernet MTU of $1500$ bytes~\cite{RFC_MTU}. Under this setting, only the \emph{Service Request} and \emph{Spectrum Request} messages experience fragmentation, increasing from one to two packets per message. The lower plot in Figure~\ref{fig:Fragmentation-Analysis} further analyzes communication overhead under varying MTU sizes, illustrating the proportion of protocol and network headers relative to payload size. As expected, messages with smaller payloads incur higher relative overhead due to header dominance.

\begin{figure}[ht]
    \centering
    \includegraphics[width=\linewidth]{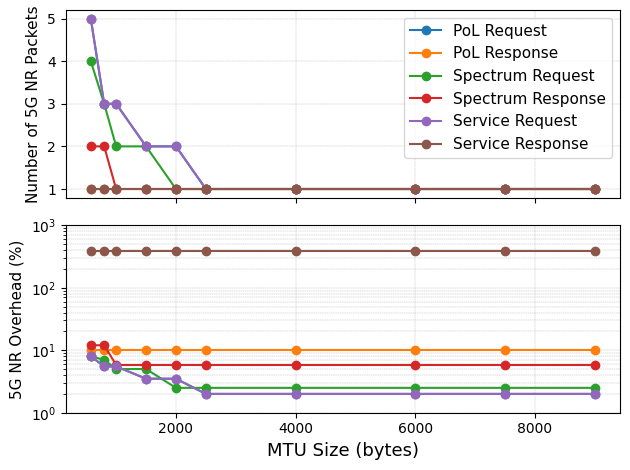}
    \caption{Fragmentation analysis for SLAPX protocol messages with varying MTU sizes} 
    \label{fig:Fragmentation-Analysis} 
\end{figure}

\textbf{DoS Experiments:} Each of the three attack cases, along with a baseline scenario without SLAPX, is evaluated independently. For each scenario, we simulate 10 seconds of network time, with the attack initiated at $t=2$ and terminated at $t=8$.


\textit{Baseline DoS:} To establish a baseline and expose system vulnerabilities under adversarial load, we first evaluate the network in the absence of any protection mechanism. In this baseline scenario, SLAPX is disabled and malicious UEs directly flood the server with requests to induce service disruption. We conduct a large set of experiments with varying $N_U$ and $r_{mal}$. System performance is evaluated using three key metrics: (i) the number of queued messages ($N_Q$), (ii) the average delay experienced by messages in the queue ($t_Q$), and (iii) the number of requests from benign UEs that are dropped due to full queue capacity and busy workers ($N_D$).


\begin{figure*}[ht]
    \begin{subfigure}{0.33\linewidth}
        \includegraphics[width=\linewidth]{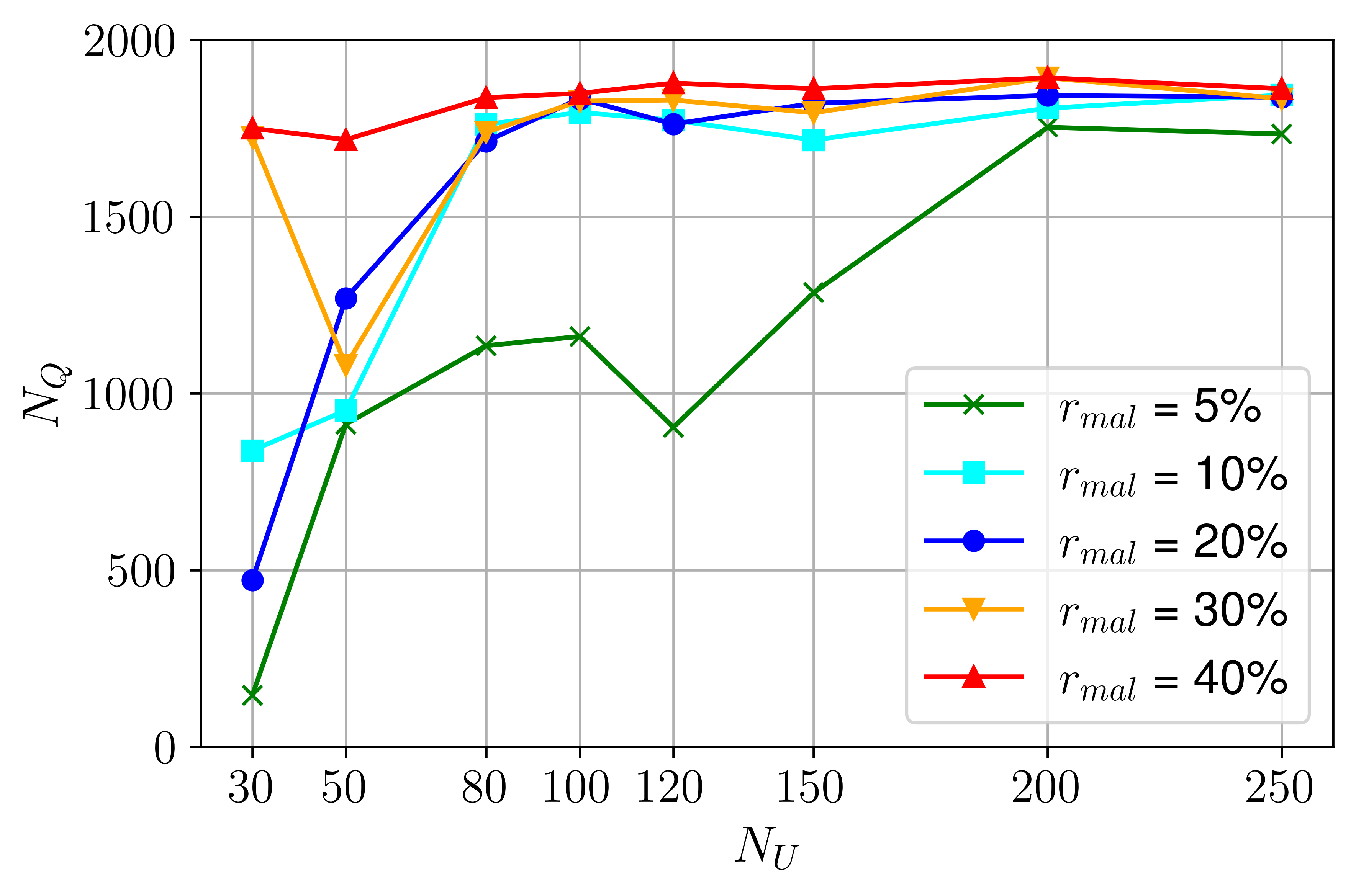}
        \caption{Total number of queued messages ($N_Q$) against total number of UEs ($N_U$)}
        \label{fig:Queued-Messages-0}
    \end{subfigure}
    \begin{subfigure}{0.33\linewidth}
        \includegraphics[width=\linewidth]{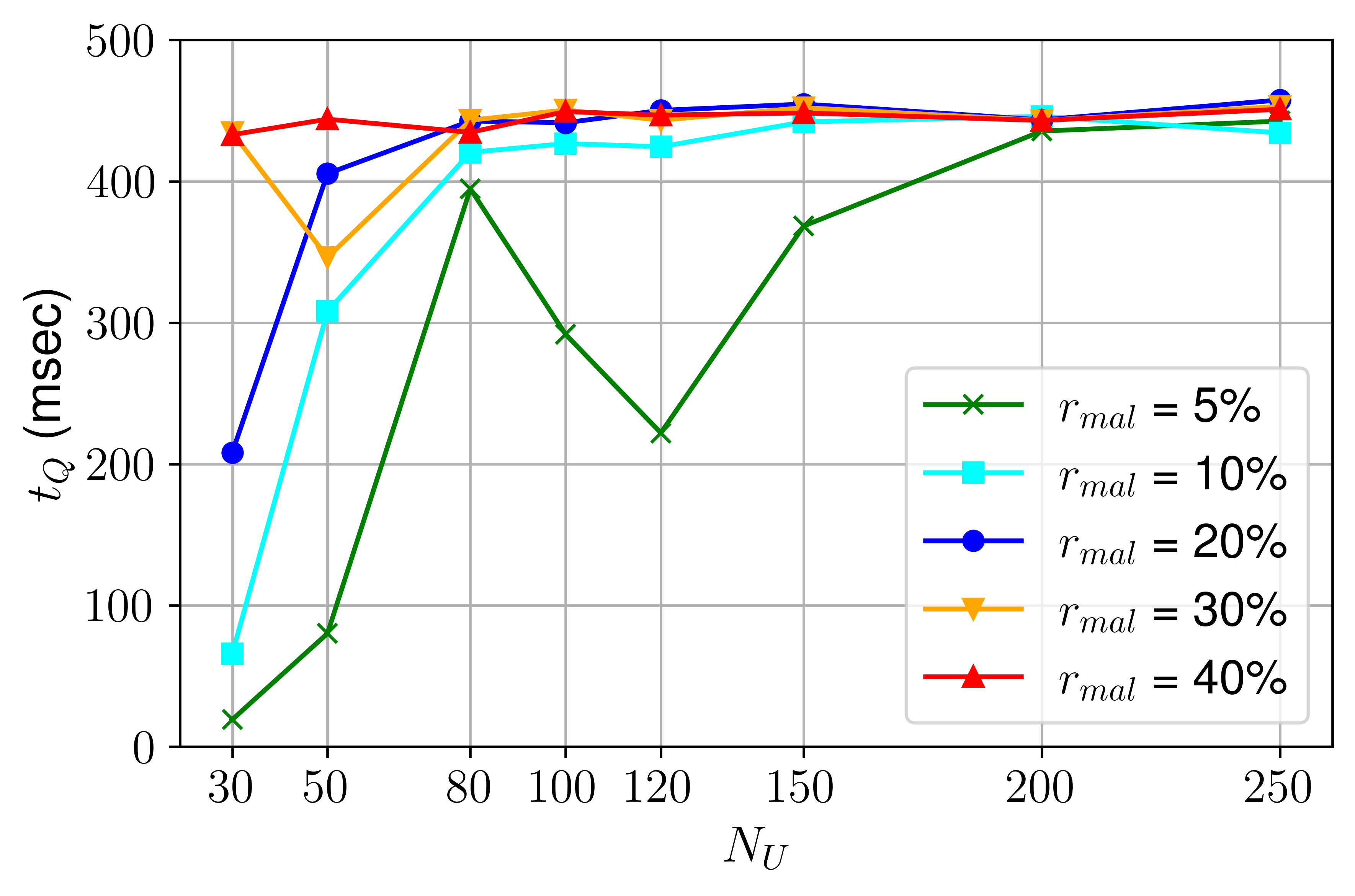}
        \caption{Average waiting time of messages in the queue ($t_Q$) against total number of UEs ($N_U$)}
        \label{fig:Avg-Waiting-Time-0}
    \end{subfigure}
    \begin{subfigure}{0.33\linewidth}
        \includegraphics[width=\linewidth]{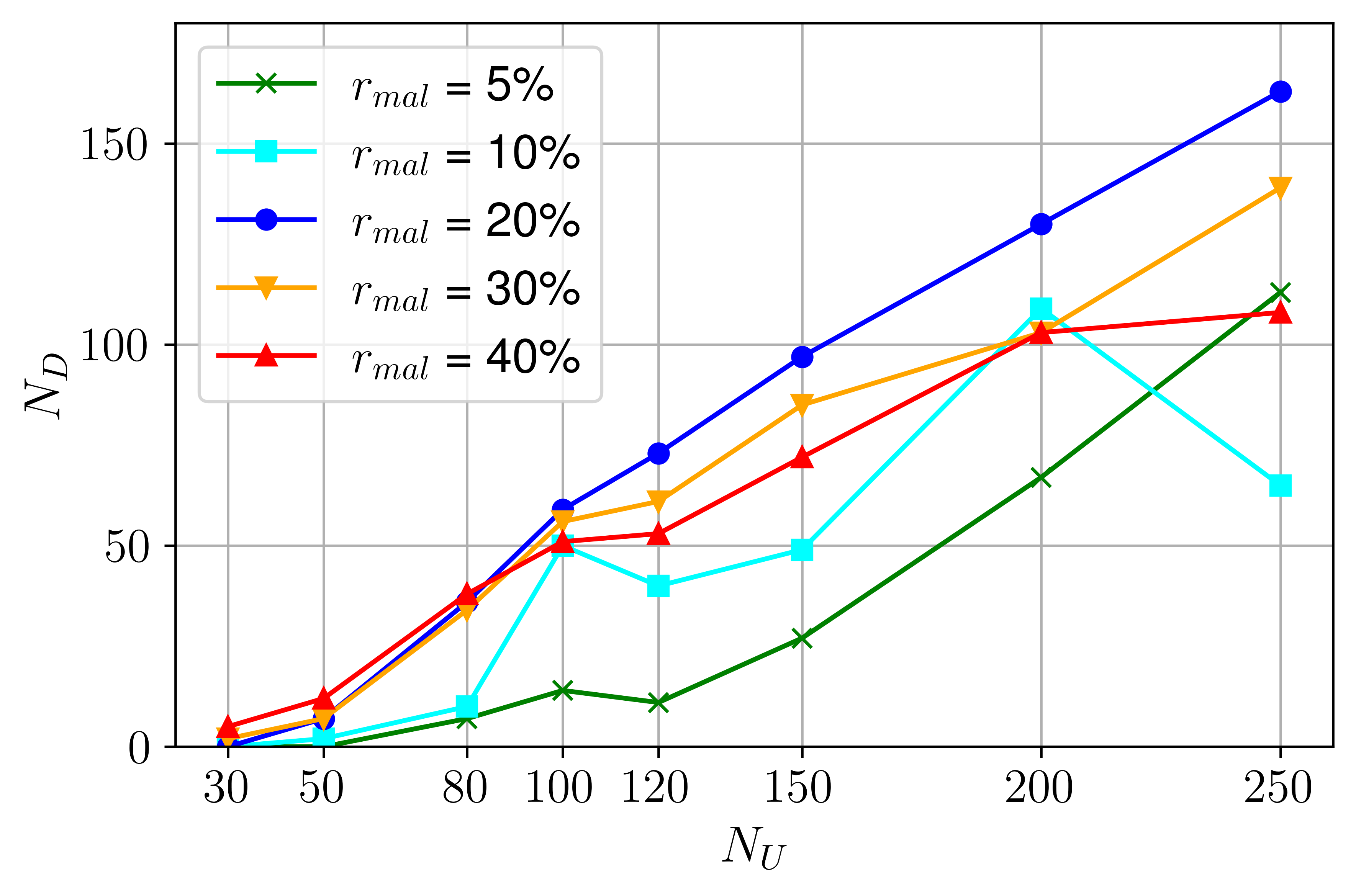}
        \caption{Total number of dropped messages ($N_D$) against total number of UEs ($N_U$)}
        \label{fig:Dropped-Messages-0}
    \end{subfigure}
    \vspace{-2em}
    \caption{Performance metrics for the Baseline DoS scenario across different numbers of UEs and malicious UE ratios}
    \label{fig:Experiments-DoS-Scenario-0}\vspace{-4mm}
\end{figure*}


Among the evaluated performance metrics, Figure~\ref{fig:Queued-Messages-0} illustrates the number of requests that cannot be processed immediately and are therefore placed into the server queue, denoted by $N_Q$. Apparently, $N_Q$ increases sharply with both $N_U$ and $r_{mal}$, reaching nearly 2000 messages in environments with high ratio of malicious UEs. Note that this value reflects the cumulative number of messages enqueued over time, rather than the maximum queue length at any instant. This distinction arises because queued messages are continuously dequeued for processing as workers become available while new requests arrive concurrently.  
Similar trends are observed in Figure~\ref{fig:Avg-Waiting-Time-0}, which reports the average waiting time of messages in the server queue, representing the time a request spends queued before being processed by an available worker and detoted by $t_Q$. Without anti-DoS countermeasures provided by SLAPX, $t_Q$ exceeds 100~ms in most scenarios and can rise to nearly 500~ms as $N_U$ and $r_{mal}$ increase. Such delays can severely impact time-sensitive applications and degrade overall service quality.



The most severe performance degradation is observed in Figure~\ref{fig:Dropped-Messages-0}, which shows the number of benign requests dropped by the server due to the DoS attack, denoted by $N_D$. As expected, $N_D$ increases with $N_U$, since a larger population generates more legitimate traffic. Notably, the highest drop rate occurs when $r_{mal}$ is 20\%, exceeding even the $r_{mal} = 40\%$ case. This behavior can be explained by the reduction in the number of benign UEs as the malicious fraction grows, which ultimately limits the volume of legitimate traffic subject to dropping. This finding displays the vulnerability of a regular SAS without any counter-DoS mechanism.


\textit{Full Protocol DoS:} 
\begin{figure*}[ht]
    \centering
    \begin{subfigure}{0.43\linewidth}
        \includegraphics[width=\linewidth]{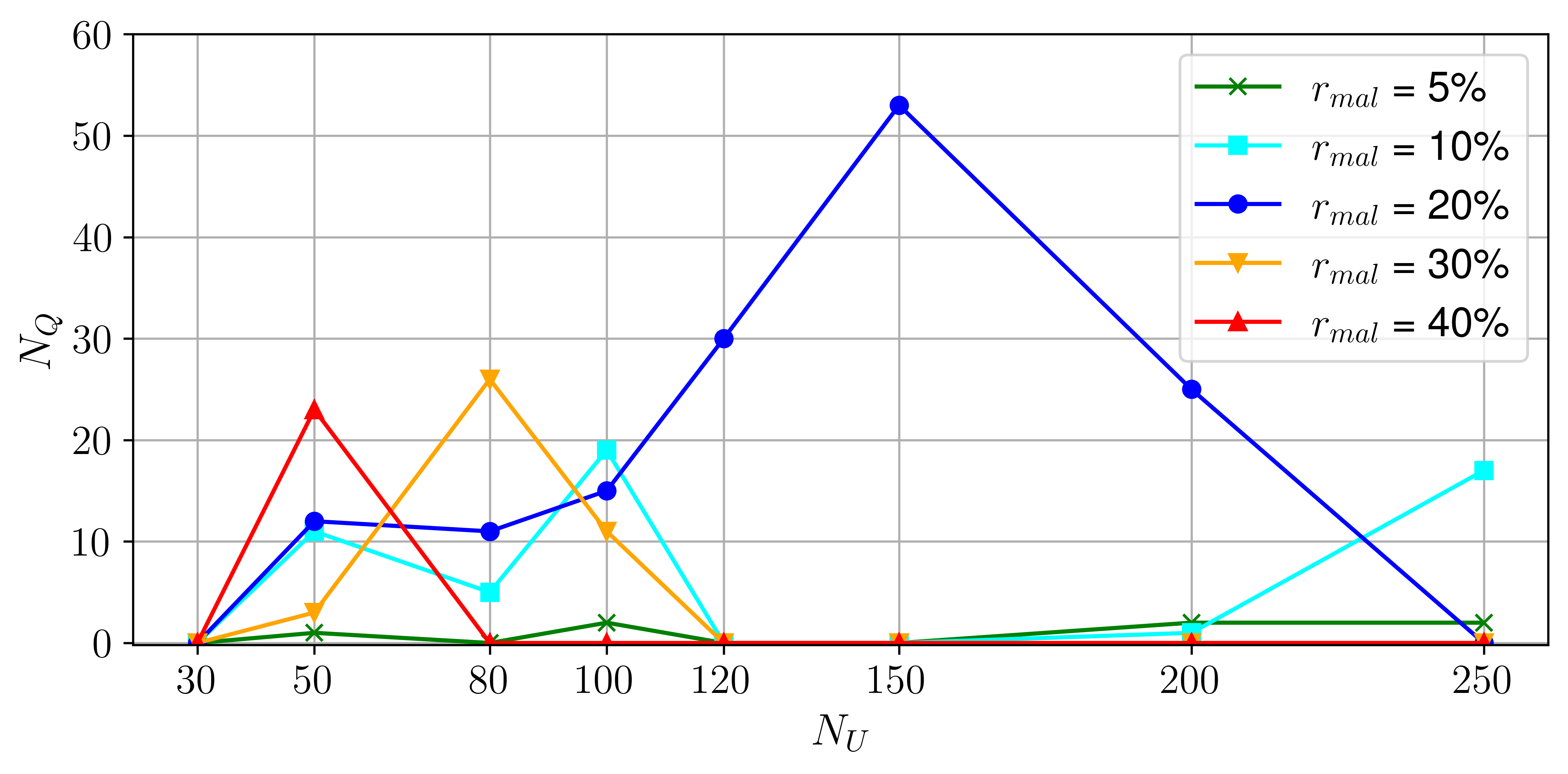}
        \caption{Total number of queued messages ($N_Q$) against total number of UEs ($N_U$)}
        \label{fig:Queued-Messages-1}
    \end{subfigure}
    \begin{subfigure}{0.43\linewidth}
        \includegraphics[width=\linewidth]{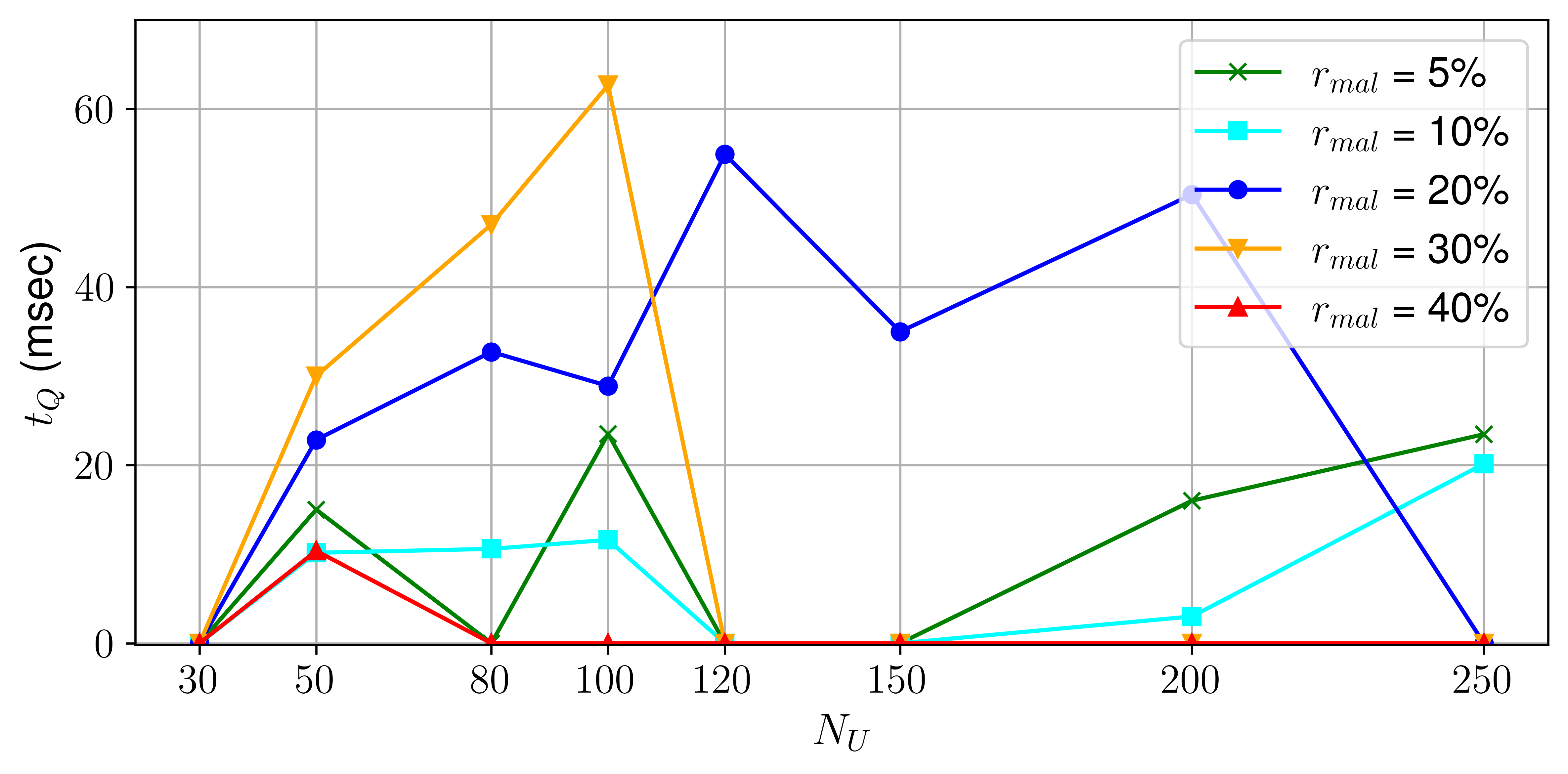}
        \caption{Average waiting time of messages in the queue ($t_Q$) against total number of UEs ($N_U$)}
        \label{fig:Avg-Waiting-Time-1}
    \end{subfigure}
    \vspace{-2mm}
    \caption{Performance metrics across different numbers of UEs and malicious UE ratios, in the \textit{Full Protocol DoS} scenario}
    \label{fig:Experiments-DoS-Scenario-1}\vspace{-5mm}
\end{figure*}
We conduct experiments with varying $N_D$ and $r_{mal}$, and summarize the results in Figure~\ref{fig:Experiments-DoS-Scenario-1}. System performance is evaluated using three metrics: $N_Q$, $t_Q$, and $N_D$. As shown in Figure~\ref{fig:Queued-Messages-1}, SLAPX decreases $N_Q$ to fewer than 50 messages during the simulations, a significant reduction compared to the baseline scenario, which yielded 1000 messages. Consistently, $t_Q$ remains below 65~ms, as illustrated in Figure~\ref{fig:Avg-Waiting-Time-1}. Moreover, no benign message drops are observed ($N_D=0$), demonstrating that the DoS attack is effectively mitigated by SLAPX under this scenario.

\textit{Computation-Bypassing DoS:}
\begin{figure*}[ht]
    \centering
    \begin{subfigure}{0.43\linewidth}
        \includegraphics[width=\linewidth]{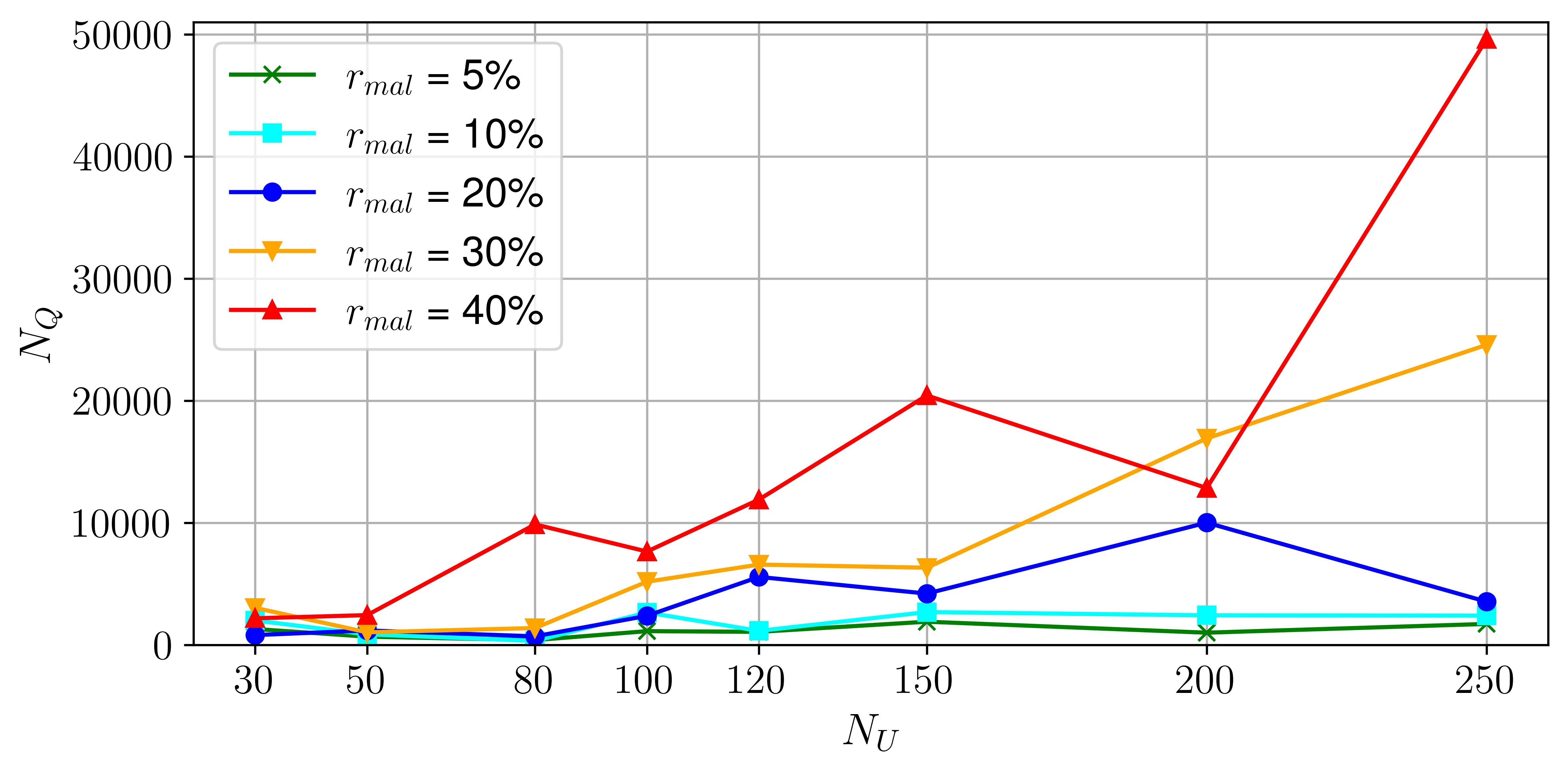}
        \caption{Total number of queued messages ($N_Q$) against total number of UEs ($N_U$)}
        \label{fig:Queued-Messages-2}
    \end{subfigure}
    \begin{subfigure}{0.43\linewidth}
        \includegraphics[width=\linewidth]{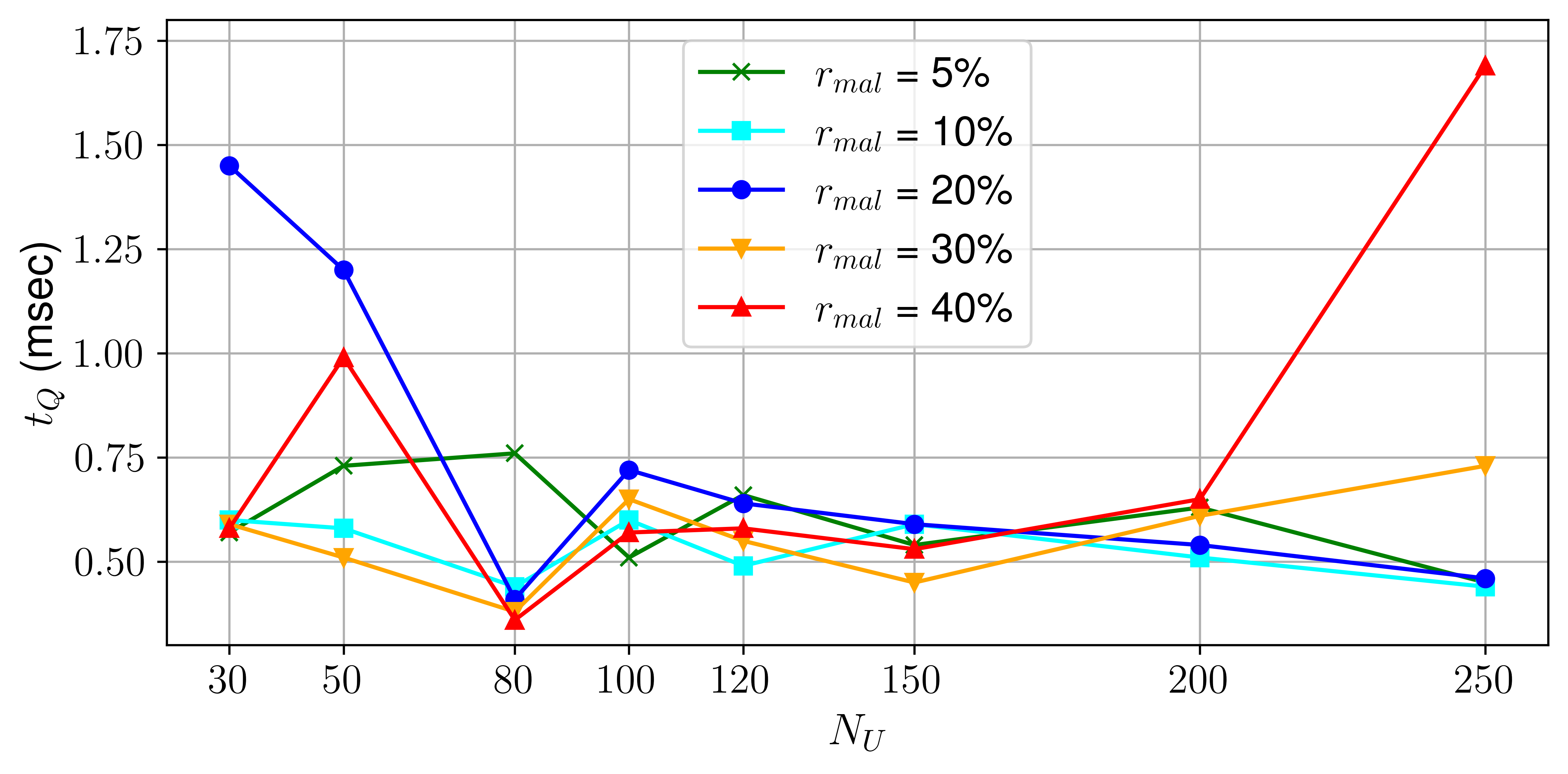}
        \caption{Average waiting time of messages in the queue ($t_Q$) against total number of UEs ($N_U$)}
        \label{fig:Avg-Waiting-Time-2}
    \end{subfigure}
    \vspace{-1em}
    \caption{Performance metrics across different numbers of UEs, in the \textit{Computation-Bypassing DoS} scenario}
    \label{fig:Experiments-DoS-Scenario-2}\vspace{-4mm}
\end{figure*}
\begin{figure*}[ht]
    \centering
    \begin{subfigure}{0.43\linewidth}
        \includegraphics[width=\linewidth]{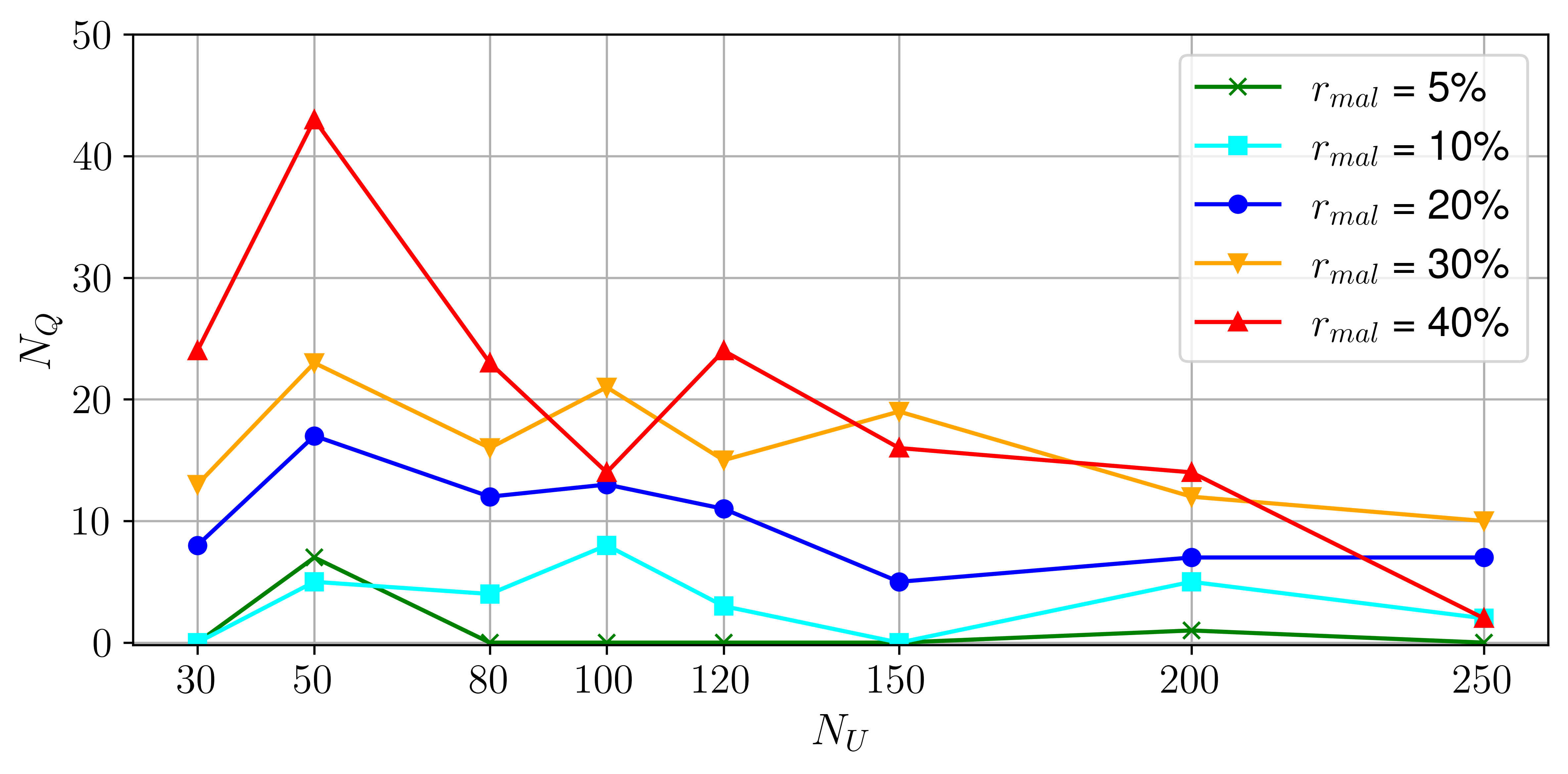}
        \caption{Total number of queued messages ($N_Q$) against total number of UEs ($N_U$)}
        \label{fig:Queued-Messages-3}
    \end{subfigure}
    \begin{subfigure}{0.43\linewidth}
        \includegraphics[width=\linewidth]{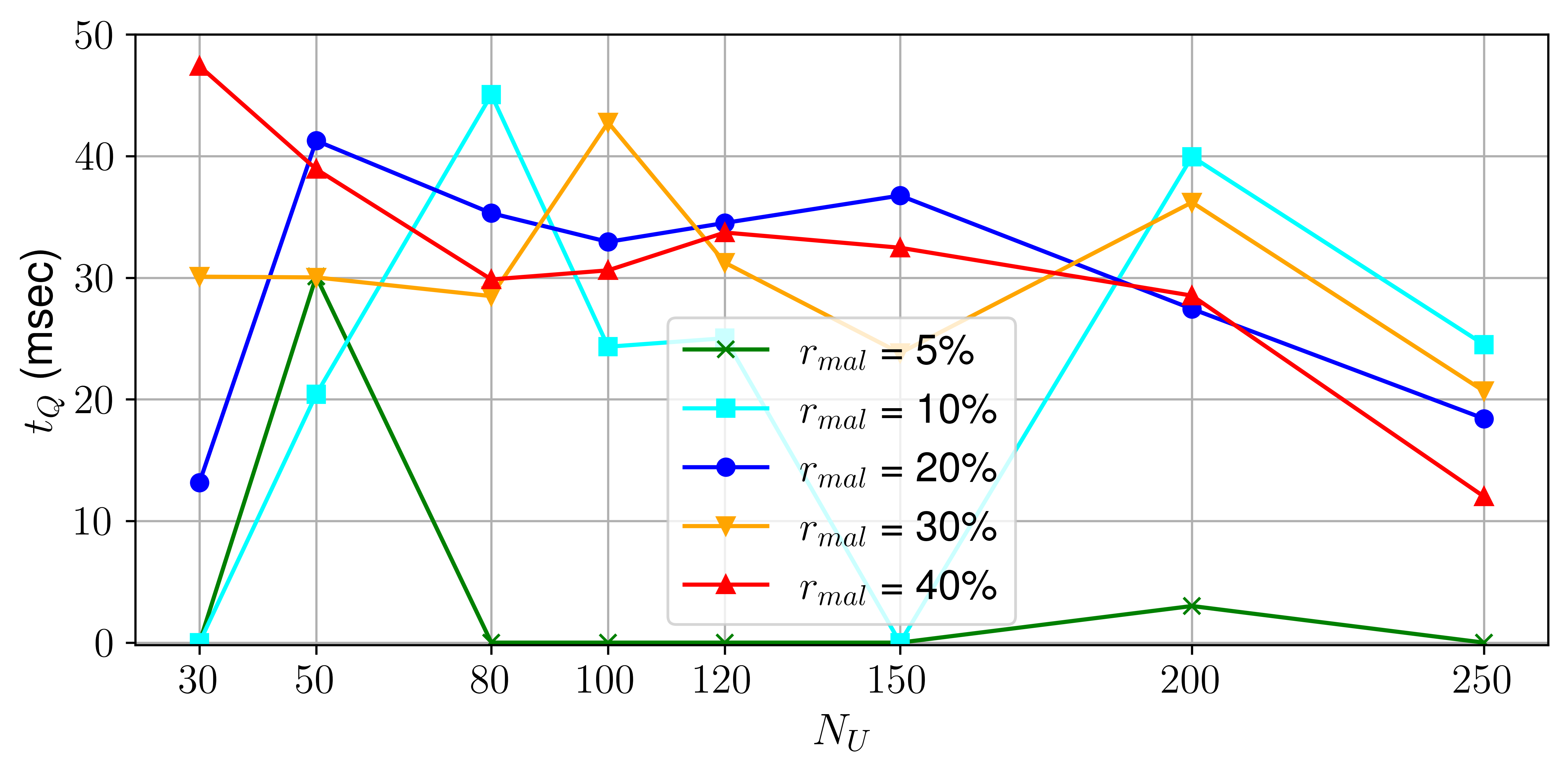}
        \caption{Average waiting time of messages in the queue ($t_Q$) against total number of UEs ($N_U$)}
        \label{fig:Avg-Waiting-Time-3}
    \end{subfigure}
    \vspace{-1em}
    \caption{Performance metrics across different numbers of UEs, in the \textit{Proof Precomputation DoS} scenario}
    \label{fig:Experiments-DoS-Scenario-3}\vspace{-4mm}
\end{figure*}
Again, we conduct experiments with varying $N_U$ and $r_{mal}$, with results shown in Figure~\ref{fig:Experiments-DoS-Scenario-2}. The system is evaluated using the same three performance metrics as in the previous scenario. As illustrated in Figure~\ref{fig:Avg-Waiting-Time-2}, $t_Q$ remains almost negligible and is even lower than in the \textit{Full-Protocol DoS} scenario, showcasing the effectiveness of SLAPX. While Figure~\ref{fig:Queued-Messages-2} shows a relatively larger $N_Q$, this increase is expected and inevitable because adversaries can send many more requests without having to perform the computational challenge. Nevertheless, in this scenario, the server promptly identifies and discards malicious requests that lack valid puzzle solutions, preventing them from consuming processing resources. The effectiveness of this early filtering is further confirmed by the consistently low $t_Q$ observed across all configurations, indicating that SLAPX does not incur significant performance degradation in this scenario either. Still, one can place a network-level firewall in front of the application and add customized rules to drop high-frequency messages, from the same source, on a network level. This additional mechanism would prevent the server queue from filling quickly and thus, helps to obtain a much lower $N_q$.



\begin{figure}
    \centering
    \includegraphics[width=0.9\linewidth]{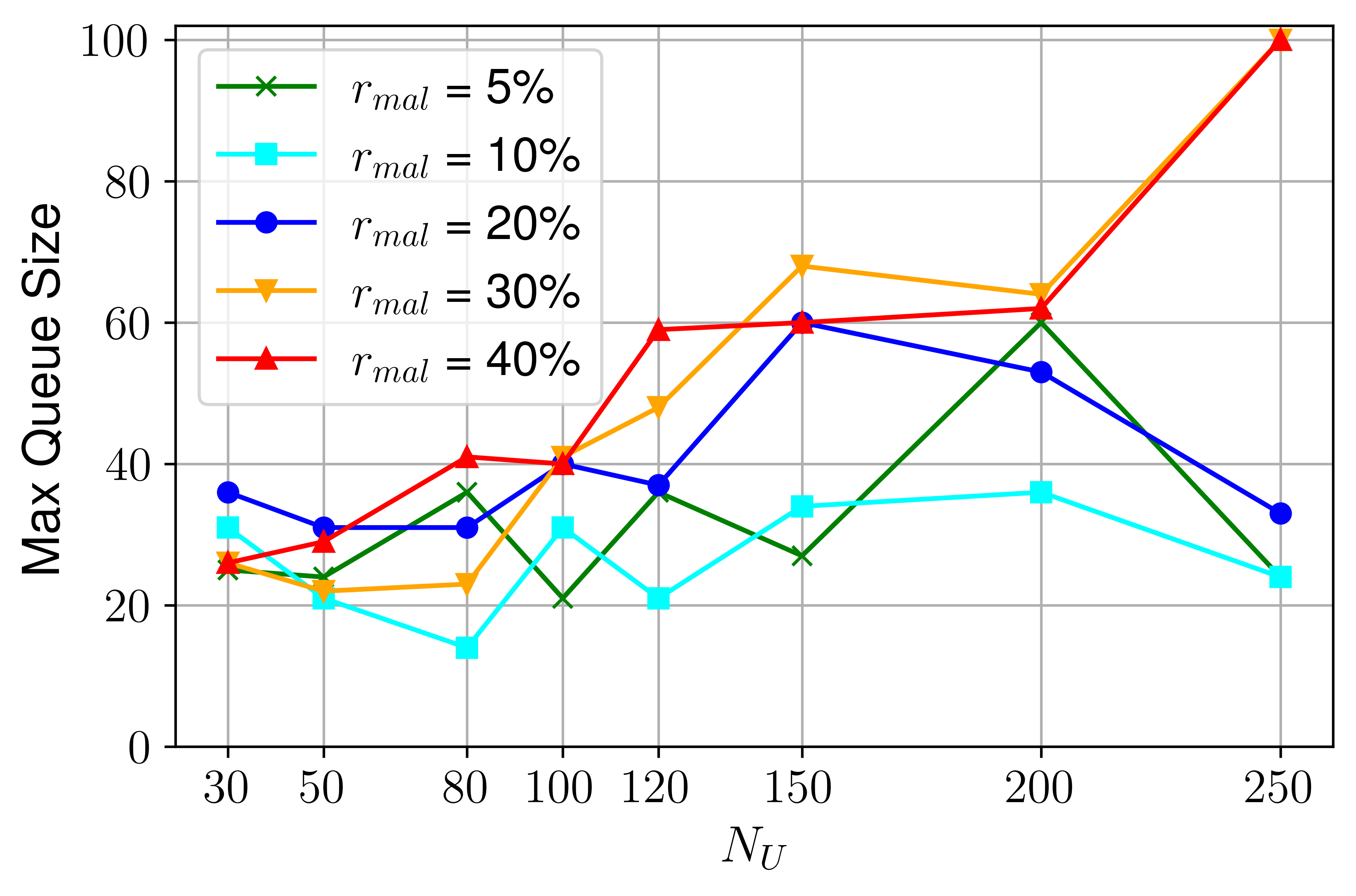}
    \vspace{-2mm}
    \caption{Maximum queue size against total number of UEs ($N_U$), in the Computation-Bypassing DoS scenario}
    \label{fig:Max-Queue-Length-2}
\end{figure}

To further explain the large number of queued messages observed in this scenario, we examine the maximum queue length reached during the simulations, as shown in Figure~\ref{fig:Max-Queue-Length-2}. This analysis was omitted for the \textit{Full-Protocol DoS} scenario due to the consistently low $N_Q$. In contrast, the present scenario exhibits a higher volume of queued requests, warranting closer inspection of queue saturation. As illustrated, the queue capacity of 100 is reached only in two configurations—both with 250 UEs and $r_{mal}$ of $30\%$ and $40\%$. Although packet drops occur in these cases, all dropped packets originate from malicious UEs, and no benign requests are denied service. This result confirms that SLAPX effectively contains the attack and preserves service availability for legitimate users.

\begin{figure}
    \centering
    \includegraphics[width=0.90\linewidth]{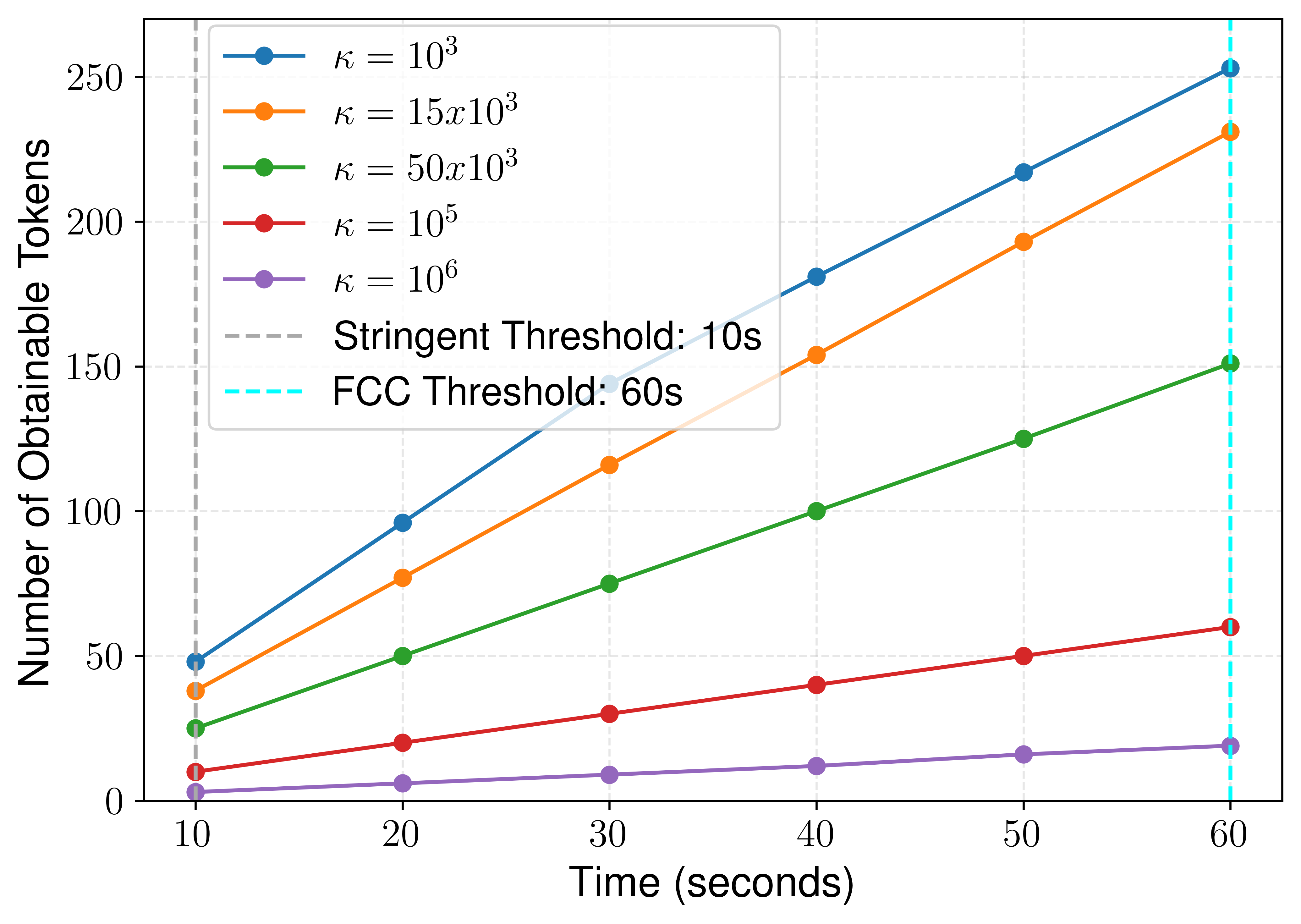}
    \caption{Limits of the proof pre-computation process}
    \label{fig:Puzzle-Pre-computation-Limits}\vspace{-1mm}
\end{figure}

\textit{Proof Precomputation DoS:}
Again, we conduct experiments with varying $N_U$ and $r_{mal}$, with results displayed in Figure~\ref{fig:Experiments-DoS-Scenario-3}. The system is evaluated using the same three performance metrics, with results shown in Figures~\ref{fig:Queued-Messages-3} and~\ref{fig:Avg-Waiting-Time-3}. As observed, system performance closely resembles that of the previous protected scenarios and remains significantly better than the baseline configuration without SLAPX. Notably, zero benign messages are dropped ($N_D = 0$), demonstrating that SLAPX successfully prevents service disruption across all three DoS attack scenarios.   

In addition to the standard metrics evaluated across all DoS scenarios, the puzzle pre-computation attack introduces an additional dimension of analysis. Although malicious UEs may attempt to pre-compute multiple puzzles in advance, the effectiveness of such an attack is inherently bounded. Each proof of location ($\pol$) includes a timestamp and is valid only within a limited time window. As a result, puzzle solutions expire after a fixed duration, and excessive pre-computation causes earlier solutions to become invalid before they can be used. Consequently, the maximum number of usable tokens an adversary can accumulate is constrained by the PoL expiration threshold and the puzzle difficulty parameters. This trade-off is illustrated in Figure~\ref{fig:Puzzle-Pre-computation-Limits}. As illustrated, no UE can pre-compute more than 250 puzzles, even under the easiest configuration and with the FCC-defined validity threshold of 60 seconds. If the system operator adopts a stricter expiration window or increases puzzle difficulty, the number of obtainable tokens drops sharply. This analysis confirms that the defense provided by SLAPX can withstand adversaries with strong computational capabilities.

\textbf{Location Spoofing Experiments:}
Here, we present the results of our experiments on the two location spoofing scenarios, each of which is evaluated separately.



\textit{Distance Hijacking:}  For one of the location spoofing scenarios, we evaluate the effectiveness of the Proximity Verification function (\textit{ProxVerify}), which enables a gNodeB to estimate the distance to UEs using a combination of received signal strength (RSS) and round-trip time (RTT) measurements. In the considered distance hijacking attack, a malicious UE communicates with a compromised benign UE over an ad-hoc WiFi network using CSMA, while the benign UE maintains a legitimate 5G NR connection with the gNodeB. The benign UE is configured to relay all messages between the malicious UE and the gNodeB without delay. The proximity threshold is set to 50 meters, in accordance with FCC guidelines~\cite{FCC_Thresholds}, such that devices within this range are classified as proximate, while those beyond it are considered distant.

\begin{figure}
    \centering
    \includegraphics[width=0.9\linewidth,trim={1.7cm 2cm 3cm 2cm},clip]{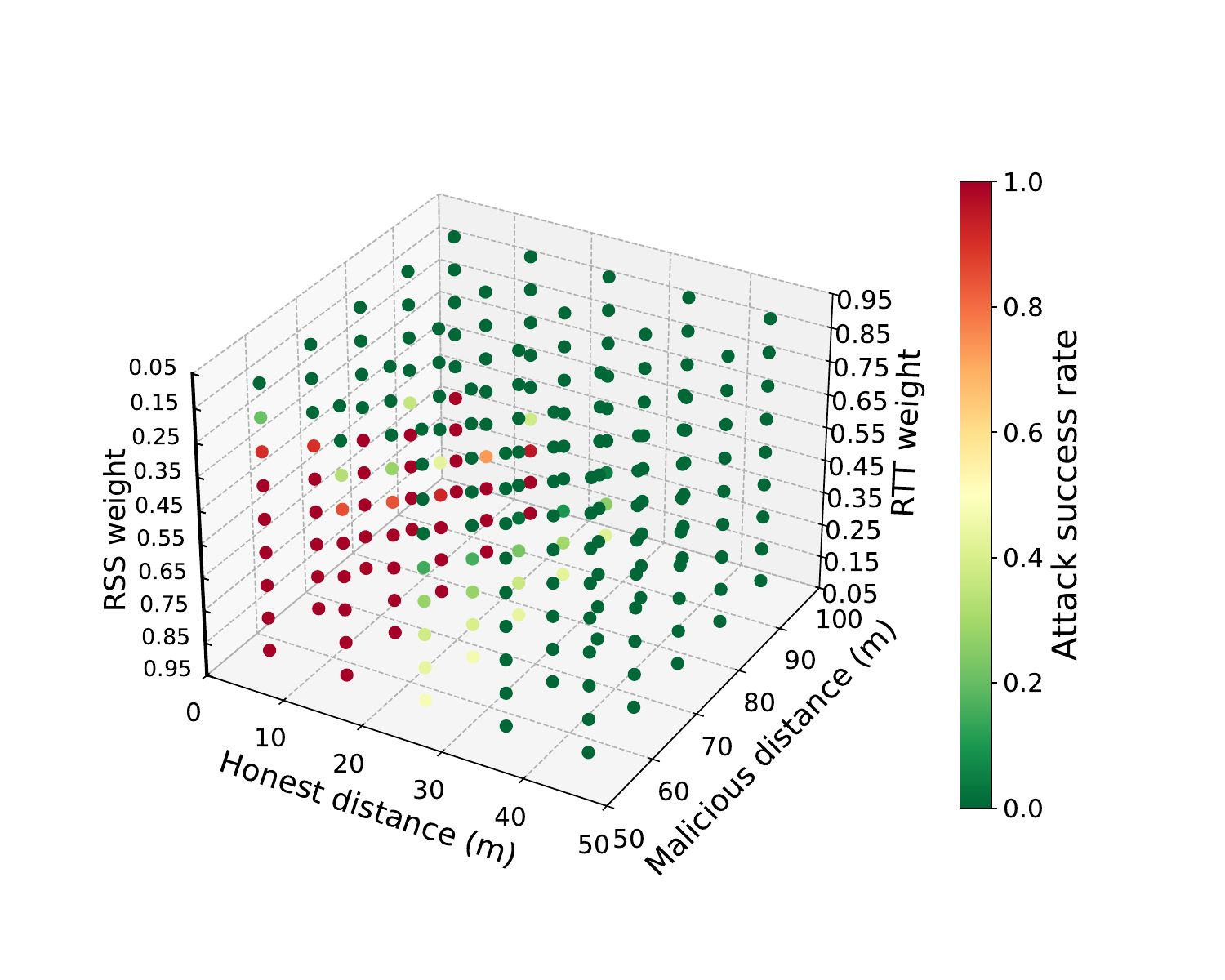}
    \caption{3D Scatter chart for attack success rate in the distance hijacking attack}
    \label{fig:Relay-3D-Scatter}
\end{figure}

\begin{table}[]
\caption{Parameters of distance hijacking simulations}
    \centering
    \resizebox{0.99\linewidth}{!}{
    \begin{tabular}{|c|c|c|} \hline
         \textbf{Parameter} & \textbf{Value Range} & \textbf{Value Step} \\ \hline
         Honest UE distance & 0 - 50 meters & 10 meters \\ \hline
         Malicious UE distance & 50 - 100 meters & 10 meters \\ \hline
         RTT weight in ProxVerify & 0.1 - 0.9 & 0.1 \\ \hline
    \end{tabular}
    }
    \label{tab:distance-highjacking-parameter-configuration}
\end{table}

To investigate the success of relay attacks under multiple factors, we vary a set of configuration parameters listed in Table~\ref{tab:distance-highjacking-parameter-configuration}. For each configuration, the experiment is repeated 100 times and the corresponding attack success rate is recorded. Figure~\ref{fig:Relay-3D-Scatter} presents a 3D scatter plot in which color encodes the attack success rate as a function of the honest UE distance to the gNodeB, the malicious UE distance, and the relative weighting of RTT versus RSS in the \textit{ProxVerify} function. The results show that attack success increases as the honest UE moves closer to the gNodeB and as the influence of RSS outweighs that of RTT. This observation aligns with earlier results, indicating that stronger emphasis on RTT improves resistance against distance hijacking. Moreover, as the honest UE approaches the gNodeB, the resulting higher RSS can be exploited by the malicious UE, further increasing the likelihood of a successful relay attack.


\textit{Distance Fraud:} In this scenario, we examine the DBP, which enables a client to authenticate using a nearby, already authenticated device when no access point is available. We consider one honest UE that is legitimately connected to the network and one malicious UE attempting to mislead the nearby device into accepting it as being closer than its actual distance. The DBP enforces proximity by executing multiple rounds of rapid challenge–response exchanges. In each round, the device must return a sequence of bits in the correct order and within a strict time bound determined by the measured RTT. 
In this setup, we initiate the distance fraud scenario by first placing an authenticated, honest UE in the network and then introducing a malicious UE at a random distance between 50 and 100 meters from it. Since it is not possible to exceed the speed of light, the malicious device cannot reduce RTT to appear closer. However, the adversary can guess the challenge bits in advance and respond before receiving them. If the guesses are correct, the malicious UE may succeed in authenticating despite being outside the allowed proximity range.



To model the DBP and characterize the attacker’s capabilities, we consider the following parameters:\vspace{-2mm} 
\begin{itemize}[leftmargin=*]
    \item \textbf{Rounds}: The total number of challenge bits transmitted by the verifier (i.e., the honest UE) to the malicious UE during the DBP execution. \vspace{-1.5mm} 
    
    \item \textbf{Tolerance}: The fraction of rounds in which a device is permitted to respond incorrectly while still being authenticated. For example, with a tolerance of 0.2 and 100 rounds, authentication succeeds if at least 80 rounds are answered correctly. \vspace{-1.5mm}
    
    
    \item \textbf{Guess Probability}: The probability that a malicious UE correctly predicts a challenge bit before receiving it. While this probability is $50\%$ under random guessing, we conservatively allow higher values to model an attacker with partial predictive advantage.
    
\end{itemize}



\begin{figure}
    \centering
    \includegraphics[width=0.85\linewidth,trim={1.5cm 0 0.7cm 0},clip]{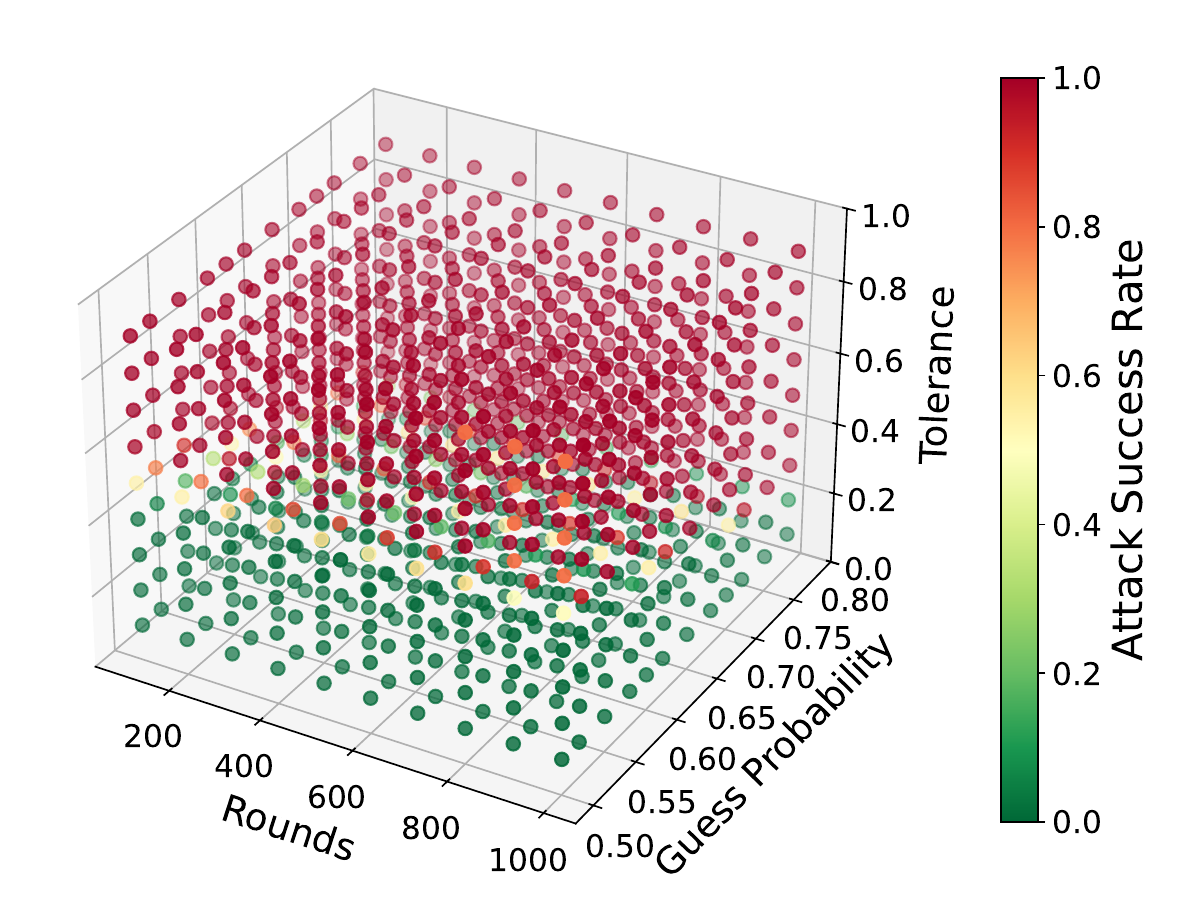}
    \caption{3D Scatter chart for attack success rate in the distance fraud attack}
    \label{fig:DistanceFraud-3D-Scatter}
\end{figure}

To evaluate the success probability of distance fraud attacks, Figure~\ref{fig:DistanceFraud-3D-Scatter} presents a 3D scatter plot in which color denotes the attack success rate as a function of the number of rounds, the attacker’s guess probability, and the DBP tolerance threshold. As expected, attack success increases with higher tolerance levels and greater guess probabilities, since both parameters favor the adversary. Nevertheless, when the tolerance is set to a reasonable value—such as $10\%$ or lower—the probability of a successful distance fraud attack becomes negligible.

\section{Related Work} \label{sec:discussion}

This section summarizes prior efforts and highlights key open problems along three major dimensions while describing our contributions beyond these efforts.

\textbf{(i) Location Privacy and Anonymity in DB-CRN:} In DB-CRNs, regulatory requirements mandate that unlicensed users (i.e., SUs) disclose precise location information, device characteristics, and spectrum requests to geolocation databases in order to obtain authorized channels and transmission parameters~\cite{agarwal2022survey}. This regulatory enforcement fundamentally conflicts with the principle of privacy, as users must sacrifice their privacy for spectrum access, thereby exposing mobility patterns and long-term behavioral profiles, which then enables tracking and profiling attacks~\cite{grissa2021anonymous}. Moreover, location privacy without anonymity or unlinkability remains insufficient, as stable identifiers combined with coarse location allow tracking and linkage across queries~\cite{newman2022spectrum}. Standard protections (e.g., TLS) do not prevent legitimate entities from observing, storing, and correlating sensitive query metadata~\cite{shi2021challenges}, either. Furthermore, the loss of anonymity expands the attack surface and can facilitate active threats such as selective DoS, targeted blocking, and user discrimination~\cite{weifang2010denial, darzi2024privacy}. 

Despite extensive research efforts, existing privacy-enhancing mechanisms remain inadequate for real-world DB-CRN deployments. Schemes relying on k-anonymity or pseudonymous identifiers~\cite{zhu2019lightweight} provide only weak protection and lack rigorous security guarantees unless very large anonymity sets are maintained—an assumption that rarely holds in dense, mobile, or highly dynamic environments. Techniques based on location obfuscation or differential privacy~\cite{ul2022differential}, while theoretically grounded, intentionally perturb location data and consequently impair the precision and dependability of spectrum availability decisions. Cryptographic approaches leveraging Private Information Retrieval (PIR)~\cite{grissa2021anonymous, xin2016privacy} offer stronger privacy guarantees but incur prohibitive computational and communication overhead, hindering scalability and practical deployment, while typically assuming benign user behavior and neglecting adversarial users.
Consequently, our proposed solution utilizes primitives such as DAC and RLRS to provide effective protection in DB-CRNs through lightweight and scalable mechanisms that jointly ensure location privacy, anonymity, and unlinkability, while remaining compatible with regulatory mandates and operating in resource-constrained environments. 


\textbf{(ii) Location Proof and Spoofing Resistance in DB-CRN:}  
DB-CRNs fundamentally operate as location-based systems in which spectrum availability decisions depend on real-time, user-provided location information stored in geolocation DBs. Consequently, 
spoofed or falsified locations undermine the basis of regulatory authorization and can cause disproportionate harm, as attackers may impersonate legitimate entities or manipulate location and usage data to obtain unauthorized spectrum access, resulting in spectrum interference, service disruption, and economic loss~\cite{nguyen2019spoofing}. 

While some approaches introduce location-proof mechanisms, they often rely on strong trust assumptions such as trusted infrastructure, specialized hardware, or dedicated location-verification servers, which are impractical in rural, ad hoc, or infrastructure-limited deployments~\cite{xin2016privacy}. The need for location verification becomes even more critical when combined with location privacy and anonymity, creating a fundamental paradox in DB-CRNs: revealing precise location information to prove correctness undermines privacy, while hiding location without verification enables spoofing. Thus, many existing solutions compromise user privacy by exposing fine-grained location data to verifiers, enabling tracking and linkage. These shortcomings highlight the necessity for a practical, privacy-preserving location verification framework that can authenticate user proximity without disclosing sensitive location data, while remaining robust against diverse spoofing and distance-fraud attacks under realistic DB-CRN deployment conditions~\cite{darzi2025slap}.

By utilizing primitives such as DBP, our proposed framework can thwart potential location spoofing attacks while preserving anonymity and unlinkability properties, which are essential for ensuring the privacy of users.

\textbf{(iii) DoS Countermeasures and Spectrum Management for Next-Gen Network Systems:}
The centralized geolocation databases in DB-CRNs, due to their frequent interaction with massive IoT deployments, result in high-value choke points whose unavailability can disrupt spectrum access at scale~\cite{jakimoski2008denial}. 
Due to the cost asymmetry, even weak attackers can overwhelm spectrum DBs with excessive or illegitimate queries, causing high verification and processing overhead for servers, which in turn degrades spectrum coordination and reduces service availability. 

Although various countermeasures have been proposed, including intrusion detection systems (IDSs), blockchain-based access control, and cryptographic Client Puzzle Protocols (CPPs), most existing defenses remain reactive and focus primarily on attack detection rather than prevention~\cite{darzi2024counter}. Such approaches often require continuous monitoring of traffic patterns, which incurs substantial operational overhead, and remain vulnerable to adversarial evasion~\cite{doriguzzi2024flad}. 
CPP-based defenses impose computational costs on clients prior to service access~\cite{bostanov2021client}, thereby limiting the rate of malicious requests. However, at large scale, they introduce new bottlenecks related to puzzle generation, distribution, and parallelization, which can burden both servers and legitimate users~\cite{ali2020foundations}. While outsourcing puzzle generation can alleviate server-side overhead~\cite{darzi2024privacy}, it merely shifts the attack surface to spectrum databases, which remain vulnerable to query flooding. 

To mitigate the aforementioned limitations, our proposed framework leverages the concept of VDF and offloads the cost of puzzle computation to attackers instead. With this approach, a practical DB-CRN deployment can achieve the required lightweight and proactive rate-limiting mechanisms that are tightly integrated into the spectrum query workflow, enabling scalable DoS resilience while preserving availability and minimizing overhead.

\section{Conclusion and Future Work}
\label{sec:conclusion} 
DB-CRNs improve spectrum utilization but introduce fundamental security and privacy challenges stemming from mandatory location disclosure, vulnerability to location spoofing, and exposure to DoS attacks on centralized infrastructure. 
In this work, we presented $\slap$, a unified and regulation-compliant security framework that simultaneously addresses these challenges under realistic adversarial assumptions. $\slap$ enables privacy-preserving and anonymous spectrum access, provides adaptive and verifiable location proofs without revealing precise user location, and ensures robust DoS resilience through verifiable delay functions combined with RLRS-based rate limiting. Comprehensive cryptographic evaluation and network simulations demonstrate that $\slap$ achieves significantly lower latency and communication overhead than existing approaches while effectively resisting spoofing and DoS attacks. As $\slap$ currently targets classical security, future work will focus on extending its resilience to quantum-capable adversaries and providing long-term post-quantum security guarantees for DB-CRNs.

\appendix

\section{APPENDIX} \label{sec:appendix_security}
\subsection{Security Proofs} \label{sec:appendix_security}
Here, we present the formal security proof of $\slap$.

\noindent\textbf{Theorem~1.} \textit{The $\slap$ framework achieves the following security guarantees: (i) anonymous user authentication, ensured by the anonymity, soundness, and unforgeability of the underlying ZKPoK and SPSEQ-UC signature schemes; (ii) location privacy, provided by the unlinkability of credential-based signatures; (iii) verifiable user location, enforced through either the unforgeability of $\rlrs$ combined with signal-based proximity measurements or the security of public-key distance-bounding protocols and anonymous credential delegation; and (iv) resilience to DoS attacks, enabled by $\vdf$ and rate-limiting mechanisms.}
%
%
%
\begin{proof}\vspace{-1.5mm}
{\em (i)} \textit{\underline{Client Anonymity and Untraceability}:} $\slap$ ensures client anonymity and untraceability via \emph{credential unforgeability} and \emph{credential anonymity} (Definitions~\ref{def:unforgeability} and~\ref{def:anonymity}). For any PPT adversary $\mathcal{A}$ interacting with issuance, delegation, and presentation oracles, the advantage in distinguishing which of two honest users generated a valid presentation satisfies $Adv_{\mathcal{A}}^{\mathrm{anon}}(\lambda) := \left| \Pr[b'=b] - \tfrac{1}{2} \right|
\le negl(\lambda)$, where $b$ is the challenger’s hidden bit selecting the presenting user. Otherwise, $\mathcal{A}$ can be reduced to either (i) breaking the knowledge soundness of the underlying ZKPoK, (ii) distinguishing re-randomized structure-preserving signatures, or (iii) distinguishing public-key–switched presentations, yielding a distinguisher for the Decisional Diffie–Hellman (DDH) problem in the underlying group $\mathbb{G}$. In particular, the SPSEQ-UC scheme guarantees \emph{origin-hiding}, \emph{derivation-privacy}, and \emph{conversion-privacy}~\cite{fuchsbauer2019structure, mir2023practical}, implying that randomized signatures, delegated commitment vectors, and key-switched credentials are computationally indistinguishable from fresh ones. These properties compose under repeated application, ensuring that credential presentations remain anonymous and unlinkable with all but negligible probability.


{\em (ii)} \textit{\underline{Client Location Privacy}:} 
$\slap$ preserves user location privacy through \emph{credential unlinkability} (Definition~\ref{def:unlinkability}). In particular, location information is protected by the unlinkability of signature--commitment pairs generated via the SPSEQ-UC scheme~\cite{mir2023practical}. The scheme supports signature re-randomization and public-key switching, ensuring that multiple presentations of the same underlying commitment yield computationally independent transcripts. Formally, for any PPT adversary $\mathcal{A}$ interacting with the proving oracle and observing a sequence of credential showings associated with the same location attribute $\ell$, the advantage of linking two presentations satisfies $Adv_{\mathcal{A}}^{\mathrm{link}}(\lambda)
:=\Big|\Pr[b'=b]-\tfrac{1}{2}\Big|
\le negl(\lambda)$, where $b$ denotes the challenge bit selecting which of two unlinkable presentations is shown. Otherwise, $\mathcal{A}$ can be used to distinguish re-randomized SPSEQ-UC signatures or correlate switched public keys, contradicting the unlinkability of SPSEQ-UC under the generic group model. Consequently, repeated spectrum queries and location proofs in $\slap$ remain computationally indistinguishable and unlinkable to verifiers, ensuring location privacy with all but negligible probability.


{\em (iii)} \textit{\underline{Location verification and Spoofing Resistance}:} 
In the infrastructure-assisted scenario, security follows from the core guarantees of the $\rlrs$, namely unforgeability, signer anonymity (even against the key generator), linkability, non-slanderability, and revocable-iff-linked accountability, all proven in the random oracle model. Correctness and soundness follow directly from Definition~\ref{def:locationverification}. Each valid location proof embeds a unique event identifier and a hidden signer exponent $s$, whose deterministic projection $S = u_0^s$ (with $u_0 = G_0(\mathsf{event})$) serves as a cryptographically binding link tag. This tag cannot be forged, altered, or duplicated without solving the Discrete Logarithm (DL) or Decisional Diffie–Hellman (DDH) problems in group $\mathbf{G}$, or the $q$-Strong Diffie–Hellman ($q$-SDH) problem in the bilinear groups $(\mathbf{G}_1,\mathbf{G}_2)$, yielding adversarial success probability at most $negl(\lambda)$. Linkability ensures that any two signatures issued by the same signer under the same event yield identical tags, making replay or proof reuse publicly detectable; conversely, producing unlinkable signatures from the same key induces a reduction to DL with non-negligible probability. Revocation is enabled iff linkability holds: two linked signatures allow algebraic extraction of $U = u_0^{H(\id)}$, which is mapped to the signer identity via accumulator soundness. Any attempt to revoke anonymity without a valid linkage reduces to breaking the $q$-SDH assumption.

In the ND scenario, location verification security is ensured by three orthogonal mechanisms. First, the $\aka$ protocol provides session key indistinguishability under the hardness of the Diffie–Hellman and discrete logarithm problems in the random oracle model, implying $\mathsf{Adv}^{\aka}_{\mathcal{A}}(\lambda)\leq negl(\lambda)$ for any PPT adversary $\mathcal{A}$. Second, proximity is enforced through a symmetric distance-bounding protocol $\dbp$, where the verifier samples ${m\in\{0,1\}^{2n}}$ and executes $n$ rapid challenge–response rounds; any adversary outside the allowed distance succeeds with probability at most $\Pr[\mathsf{Accept}] \le (3/4)^n$, which is negligible in $n$, providing information-theoretic resistance against distance fraud, mafia fraud, and distance hijacking. Finally, delegated location proofs are unforgeable: any adversary producing a credential embedding a forged proof $\pol^\star$ can be reduced to either forging the SPSEQ-UC structure-preserving signature or violating the soundness of the underlying NIZK proof system, yielding $\Pr[\mathsf{Forge}(\pol^\star)] \leq \mathsf{Adv}^{\mathsf{EUF}}*{\mathsf{SPSEQ\mbox{-}UC}}(\lambda)+\mathsf{Adv}^{\mathsf{ZK}}_{\mathsf{NIZK}}(\lambda)\leq negl(\lambda)$ under the DDH assumption. Together, these guarantees ensure correct, privacy-preserving, and accountable location verification while providing strong resistance against spoofing attacks and collusion under realistic adversarial conditions.

{\em (iv)} \textit{\underline{Authenticated and Comprehensive Counter-DoS}:} 
In $\slap$, the authenticity of $\vdf$ tokens is enforced through the existential unforgeability of the underlying digital signature scheme $\sgn$ under chosen-message attacks (EUF-CMA). When instantiated with \texttt{ECDSA}, security reduces to the hardness of the Elliptic Curve Discrete Logarithm Problem (EC-DLP). Formally, for any PPT adversary $\mathcal{A}$ issuing at most $q(\lambda)$ signing queries, ${\Pr\!\big[\sgn.\verify(\pk,m^\star,\sigma^\star)=}$ ${1 \;\wedge\; m^\star\notin\mathcal{Q}\big]
\;\le\; Adv^{\mathrm{EUF\mbox{-}CMA}}_{\mathrm{ECDSA}}(\lambda)
\;\le\; negl(\lambda)}$, where $\mathcal{Q}$ denotes the set of messages signed by the legitimate issuer. Any successful forgery yields an $\texttt{ECDSA}$ forger and thus an EC-DLP solver with non-negligible probability, contradicting the assumption.

Server-side DoS resilience is achieved using a $\vdf$ instantiated in an RSA group $\mathbb{G}$ of unknown order (e.g., modulo $N=pq$). For input $x$ and delay parameter $T$, the prover computes $y=g^{2^T}$ together with a succinct proof $\pi$, while verification requires $O(\log T)$ time. The security of the $\vdf$ relies on: (i) \emph{inherent sequentiality}, as computing $g^{2^T}$ requires $T$ dependent squarings; (ii) \emph{non-parallelizability}, since no asymptotic speedup is possible in groups of unknown order; (iii) \emph{soundness}, where producing a valid proof $(y^\ast,\pi^\ast)$ with $y^\ast\neq g^{2^T}$ implies breaking the Strong RSA assumption; and (iv) \emph{uniqueness}, ensuring a single valid output for fixed $(x,T)$. 
Let $\tau$ denote the per-request $\vdf$ evaluation cost. For any PPT adversary $\mathcal{A}$ issuing at most $q(\lambda)$ requests, $\Pr\![\exists\, (y^\ast,\pi^\ast): \verify(x,T,y^\ast,\pi^\ast)=1 \ \wedge\ \mathsf{Cost}(\mathcal{A})<q($ $\lambda)\cdot\tau]
\le negl(\lambda)$, since otherwise $\mathcal{A}$ computes valid $\vdf$ outputs without incurring $\Omega(\tau)$ sequential work per accepted request, yielding a Strong RSA solver. Thus, accepted requests impose a provable per-request computational cost on clients while keeping verification efficient, establishing strong asymmetric DoS resistance.

To protect $\psd$s against query flooding, $\slap$ enforces rate limiting via event-scoped linkability provided by $\rlrs$ signatures embedded in location proofs. Each valid signature implicitly contains a secret exponent $s$ (hidden via zero-knowledge proofs) and a deterministic link tag $S=u_0^{\,s}$, where $u_0=G_0(\mathrm{event})$. The tag $S$ is deterministic, event-specific, and signer-unique. For any PPT adversary $\mathcal{A}$, $\Pr\![\exists\, \sigma_1,\sigma_2:\ \verify(\sigma_1)=\verify(\sigma_2)=1 \ \wedge\ S(\sigma_1)$ $=S(\sigma_2)\ \wedge\ s_1\neq s_2]
\le negl(\lambda)$, since $u_0^{s_1}=u_0^{s_2}$ with $s_1\neq s_2$ yields a solution to the Discrete Logarithm problem. Hence, multiple signatures issued by the same signer under the same event are publicly linkable, bounding the number of valid queries by the number of distinct ring members (e.g., APs) in the region. Moreover, \emph{revocable-iff-linked} accountability is guaranteed: only when two signatures are linked can the signer’s identity be extracted via $U=u_0^{H(\id)}$, with correctness ensured by accumulator soundness. Any attempt to revoke anonymity without linkability reduces to breaking the $q$-Strong Diffie–Hellman (q-SDH) assumption. 

By combining EUF-CMA authenticated $\vdf$ tokens, $\vdf$-enforced sequential computational cost, and $\rlrs$-based event-scoped rate limiting, $\slap$ ensures authenticated request handling, provable per-request hardness, enforced rate limits, and cryptographic accountability. Any adversary that violates availability without incurring proportional cost would yield a solver for EC-DLP, Strong RSA, or DL/DDH/q-SDH with non-negligible probability, contradicting standard hardness assumptions.
\end{proof}

\printcredits

\section*{Acknowledgment}
This work has received co-funding from the NSF and SNSF, in the framework of the SATUQ project under SNSF (Grant No 10000409) and NSF (ECCS 2444615).


\bibliographystyle{cas-model2-names}  
\bibliography{SalehRef}             

@inproceedings{caleffi2014database,
  title={Database access strategy for TV white space cognitive radio networks},
  author={Caleffi, Marcello and Cacciapuoti, Angela Sara},
  booktitle={2014 Eleventh Annual IEEE International Conference on Sensing, Communication, and Networking Workshops (SECON Workshops)},
  pages={34--38},
  year={2014},
  organization={IEEE}
}

@misc{das2015rfc,
  title={RFC 7545: Protocol to Access White-Space (PAWS) Databases},
  author={Das, S and Zhu, L and Malyar, J and McCann, P},
  year={2015},
  publisher={RFC Editor}
}

@techreport{chen2015protocol,
  title={Protocol to access white-space (paws) databases},
  author={Chen, V and Das, S and Zhu, L and Malyar, J and McCann, P},
  year={2015}
}

@article{order2023federal,
  title={Federal Communications Commission},
  author={Order, Declaratory},
  year={2023}
}

@article{zeng2019efficient,
  title={An efficient privacy-preserving protocol for database-driven cognitive radio networks},
  author={Zeng, Yali and Xu, Li and Yang, Xu and Yi, Xun},
  journal={Ad hoc networks},
  volume={90},
  pages={101739},
  year={2019},
  publisher={Elsevier}
}

@inproceedings{vaudenay2015private,
  title={Private and secure public-key distance bounding: application to NFC payment},
  author={Vaudenay, Serge},
  booktitle={International Conference on Financial Cryptography and Data Security},
  pages={207--216},
  year={2015},
  organization={Springer}
}

@inproceedings{jakimoski2008denial,
  title={Denial-of-service attacks on dynamic spectrum access networks},
  author={Jakimoski, G and Subbalakshmi, KP},
  booktitle={IEEE International Conference on Communications Workshops},
  pages={524--528},
  year={2008},
  organization={IEEE}
}

@article{abuyaghi2024positioning,
  title={Positioning in 5g networks: Emerging techniques, use cases, and challenges},
  author={Abuyaghi, Mohammad and Si-Mohammed, Samir and Shaker, George and Rosenberg, Catherine},
  journal={IEEE Internet of Things Journal},
  year={2024},
  publisher={IEEE}
}

@article{lopez2019ieee,
  title={IEEE 802.11 be extremely high throughput: The next generation of Wi-Fi technology beyond 802.11 ax},
  author={Lopez-Perez, David and Garcia-Rodriguez, Adrian and Galati-Giordano, Lorenzo and Kasslin, Mika and Doppler, Klaus},
  journal={IEEE Communications Magazine},
  volume={57},
  number={9},
  pages={113--119},
  year={2019},
  publisher={IEEE}
}

@inproceedings{rosler2025improving,
  title={Improving WiFi Ranging Through Frequency Diversity and Mobility},
  author={R{\"o}sler, Sascha and Eising, Nils and Zubow, Anatolij and Dressler, Falko},
  booktitle={2025 IEEE 50th Conference on Local Computer Networks (LCN)},
  pages={1--7},
  year={2025},
  organization={IEEE}
}

@inproceedings{gao2012location,
  title={Location privacy leaking from spectrum utilization information in database-driven cognitive radio network},
  author={Gao, Zhaoyu and Zhu, Haojin and Liu, Yao and Li, Muyuan and Cao, Zhenfu},
  booktitle={Proceedings of the 2012 ACM conference on Computer and communications security},
  pages={1025--1027},
  year={2012}
}

@inproceedings{bahrak2014protecting,
  title={Protecting the primary users' operational privacy in spectrum sharing},
  author={Bahrak, Behnam and Bhattarai, Sudeep and Ullah, Abid and Park, Jung-Min Jerry and Reed, Jeffery and Gurney, David},
  booktitle={2014 IEEE International Symposium on Dynamic Spectrum Access Networks (DYSPAN)},
  pages={236--247},
  year={2014},
  organization={IEEE}
}

@inproceedings{boneh2018verifiable,
  title={Verifiable delay functions},
  author={Boneh, Dan and Bonneau, Joseph and B{\"u}nz, Benedikt and Fisch, Ben},
  booktitle={Annual international cryptology conference},
  pages={757--788},
  year={2018},
  organization={Springer}
}

@article{wesolowski2020efficient,
  title={Efficient verifiable delay functions},
  author={Wesolowski, Benjamin},
  journal={Journal of Cryptology},
  volume={33},
  number={4},
  pages={2113--2147},
  year={2020},
  publisher={Springer}
}

@article{au2013secure,
  title={Secure ID-based linkable and revocable-iff-linked ring signature with constant-size construction},
  author={Au, Man Ho and Liu, Joseph K and Susilo, Willy and Yuen, Tsz Hon},
  journal={Theoretical Computer Science},
  volume={469},
  pages={1--14},
  year={2013},
  publisher={Elsevier}
}

@article{darzi2025slap,
  title={SLAP: Secure Location-proof and Anonymous Privacy-preserving Spectrum Access},
  author={Darzi, Saleh and Yavuz, Attila A},
  journal={2025 Silicon Valley Cybersecurity Conference (SVCC)},
pages = {1-8},
  year={2025},
  publisher = {IEEE}
}

@article{mir2023practical,
  title={Practical delegatable anonymous credentials from equivalence class signatures},
  author={Mir, Omid and Slamanig, Daniel and Bauer, Balthazar and Mayrhofer, Ren{\'e}},
  journal={Proceedings on Privacy Enhancing Technologies},
  year={2023}
}

@article{grissa2021anonymous,
  title={Anonymous dynamic spectrum access and sharing mechanisms for the CBRS band},
  author={Grissa, Mohamed and Yavuz, Attila Altay and Hamdaoui, Bechir and Tirupathi, Chittibabu},
  journal={IEEE Access},
  volume={9},
  pages={33860--33879},
  year={2021},
  publisher={IEEE}
}

@inproceedings{grissa2016efficient,
  title={An efficient technique for protecting location privacy of cooperative spectrum sensing users},
  author={Grissa, Mohamed and Yavuz, Attila and Hamdaoui, Bechir},
  booktitle={2016 IEEE conference on computer communications workshops (INFOCOM WKSHPS)},
  pages={915--920},
  year={2016},
  organization={IEEE}
}

@article{ali2020foundations,
  title={Foundations, properties, and security applications of puzzles: A survey},
  author={Ali, Isra Mohamed and Caprolu, Maurantonio and Pietro, Roberto Di},
  journal={ACM Computing Surveys (CSUR)},
  volume={53},
  number={4},
  pages={1--38},
  year={2020},
  publisher={ACM New York, NY, USA}
}

@article{chakraborty2023capow,
  title={Capow: Context-aware ai-assisted proof of work based ddos defense},
  author={Chakraborty, Trisha and Mitra, Shaswata and Mittal, Sudip},
  journal={arXiv preprint},
  year={2023}
}

@article{bostanov2021client,
  title={Client Puzzle Protocols as Countermeasure Against Automated Threats to Web Applications},
  author={Bostanov, Vladimir},
  journal={IEEE Access},
  volume={9},
  year={2021},
  publisher={IEEE}
}

@article{doriguzzi2024flad,
  title={FLAD: adaptive federated learning for DDoS attack detection},
  author={Doriguzzi-Corin, Roberto and Siracusa, Domenico},
  journal={Computers \& Security},
  volume={137},
  year={2024},
  publisher={Elsevier}
}

@article{li2015privacy,
  title={Privacy-preserving location proof for securing large-scale database-driven cognitive radio networks},
  author={Li, Yi and Zhou, Lu and Zhu, Haojin and Sun, Limin},
  journal={IEEE Internet of Things Journal},
  volume={3},
  number={4},
  pages={563--571},
  year={2015},
  publisher={IEEE}
}

@article{zhu2019lightweight,
  title={Lightweight Privacy Preservation for Securing Large-Scale Database-Driven Cognitive Radio Networks with Location Verification},
  author={Zhu, Rui and Xu, Li and Zeng, Yali and Yi, Xun},
  journal={Security and Communication Networks},
  year={2019},
  publisher={Wiley Online Library}
}

@inproceedings{troja2014leveraging,
  title={Leveraging p2p interactions for efficient location privacy in database-driven dynamic spectrum access},
  author={Troja, Erald and Bakiras, Spiridon},
  booktitle={22nd ACM SIGSPATIA Conference},
  pages={453--456},
  year={2014}
}

@inproceedings{darzi2024privacy,
  title={Privacy-Preserving and Post-Quantum Counter Denial of Service Framework for Wireless Networks},
  author={Darzi, Saleh and Yavuz, Attila Altay},
  booktitle={IEEE Military Communications Conference (MILCOM)},
  pages={1076--1081},
  year={2024},
  organization={IEEE}
}

@inproceedings{xin2016privacy,
  title={Privacy-preserving spectrum query with location proofs in database-driven CRNs},
  author={Xin, Jiajun and Li, Ming and Luo, Changqing and Li, Pan},
  booktitle={2016 IEEE Global Communications Conference (GLOBECOM)},
  pages={1--6},
  year={2016},
  organization={IEEE}
}

@article{grissa2019location,
  title={Location privacy in cognitive radios with multi-server private information retrieval},
  author={Grissa, Mohamed and Yavuz, Attila Altay},
  journal={IEEE Transactions on Cognitive Communications and Networking},
  year={2019},
  publisher={IEEE}
}

@inproceedings{darzi2024counter,
  title={Counter denial of service for next-generation networks within the artificial intelligence and post-quantum era},
  author={Darzi, Saleh and Yavuz, Attila A},
  booktitle={IEEE 6th International Conference on Trust, Privacy and Security in Intelligent Systems, and Applications},
  pages={138--147},
  year={2024},
  organization={IEEE}
}

@article{shi2021challenges,
  title={Challenges and new directions in securing spectrum access systems},
  author={Shi, Shanghao and Xiao, Yang and Lou, Wenjing and Wang, Chonggang and Li, Xu and Hou, Y Thomas and Reed, Jeffrey H},
  journal={IEEE Internet of Things Journal},
  volume={8},
  number={8},
  pages={6498--6518},
  year={2021},
  publisher={IEEE}
}

@inproceedings{kilincc2016efficient,
  title={Efficient public-key distance bounding protocol},
  author={K{\i}l{\i}n{\c{c}}, Handan and Vaudenay, Serge},
  booktitle={22nd International Conference on the Theory and Application of Cryptology and Information Security},
  pages={873--901},
  organization={Springer}
}

@article{fuchsbauer2019structure,
  title={Structure-preserving signatures on equivalence classes and constant-size anonymous credentials},
  author={Fuchsbauer, Georg and Hanser, Christian and Slamanig, Daniel},
  journal={Journal of Cryptology},
  volume={32},
  pages={498--546},
  year={2019},
  publisher={Springer}
}

@article{nguyen2019spoofing,
  title={Spoofing attack and surveillance game in geo-location database driven spectrum sharing},
  author={Nguyen-Thanh, Nhan and Ta, Duc-Tuyen and Nguyen, Van-Tam},
  journal={IET Communications},
  volume={13},
  number={1},
  pages={74--84},
  year={2019},
  publisher={Wiley Online Library}
}

@article{ul2022differential,
  title={Differential privacy in cognitive radio networks: a comprehensive survey},
  author={Ul Hassan, Muneeb and Rehmani, Mubashir Husain and Rehan, Maaz and Chen, Jinjun},
  journal={Cognitive Computation},
  volume={14},
  number={2},
  pages={475--510},
  year={2022},
  publisher={Springer}
}

@inproceedings{jasim2021cognitive,
  title={Cognitive radio network: Security and reliability trade-off-status, challenges, and future trend},
  author={Jasim, Doaa K and Sadkhan, Sattar B},
  booktitle={2021 1st Babylon International Conference on Information Technology and Science (BICITS)},
  pages={149--153},
  year={2021},
  organization={IEEE}
}

@article{agarwal2022survey,
  title={A survey on citizens broadband radio service (cbrs)},
  author={Agarwal, Pranay and Manekiya, Mohammedhusen and Ahmad, Tahir},
  journal={Electronics},
  volume={11},
  number={23},
  year={2022},
  publisher={MDPI}
}

@misc{openssl,
	title = {OpenSSL Library},
	howpublished = {\url{https://www.openssl.org/}},
	note = {Accessed: April, 2024},
	year={2024}
}

@inproceedings{weifang2010denial,
  title={Denial of service attacks in cognitive radio networks},
  author={Weifang, Wang},
  booktitle={2010 The 2nd Conference on Environmental Science and Information Application Technology},
  volume={2},
  pages={530--533},
  year={2010},
  organization={IEEE}
}

@incollection{thippeswamy2016physical,
  title={Physical layer, data link layer, network layer, transport layer, and application layer in cognitive radio networks},
  author={Thippeswamy, MN and Prasanna, A Dinesh and Takawira, F},
  booktitle={Introduction to Cognitive Radio Networks and Applications},
  pages={187--200},
  year={2016},
  publisher={Chapman and Hall/CRC}
}

@article{liang2011cognitive,
  title={Cognitive radio networking and communications: An overview},
  author={Liang, Ying-Chang and Chen, Kwang-Cheng and Li, Geoffrey Ye and Mahonen, Petri},
  journal={IEEE transactions on vehicular technology},
  volume={60},
  number={7},
  pages={3386--3407},
  year={2011},
  publisher={IEEE}
}

@article{newman2022spectrum,
  title={Spectrum: High-bandwidth anonymous broadcast with malicious security},
  author={Newman, Zachary and Servan-Schreiber, Sacha and Devadas, Srinivas},
  year={2022},
  publisher={The USENIX Association}
}

@article{Wang2007,
    author = {Wang, Yang and Lin, Chuang and Li, Quan-Lin and Fang, Yuguang},
    title = {A queueing analysis for the denial of service (DoS) attacks in computer networks},
    year = {2007},
    issue_date = {August, 2007},
    publisher = {Elsevier North-Holland, Inc.},
    address = {USA},
    volume = {51},
    number = {12},
    issn = {1389-1286},
    journal = {Comput. Netw.},
    month = aug,
    pages = {3564–3573},
    numpages = {10}
}

@misc{3GPPChannelModel,
    title={Study on channel model for frequencies from 0.5 to 100 {GHz}},
    year={2025},
    howpublished={\url{https://portal.3gpp.org/desktopmodules/Specifications/SpecificationDetails.aspx?specificationId=3173}},
    note={Accessed: 2025, October}
}

@misc{RFC_MTU,
    series =    {Request for Comments},
    number =    4638,
    howpublished =  {RFC 4638},
    publisher = {RFC Editor},
    author =    {Mike Duckett and Jerome Moisand and Tom Anschutz and Diamantis Kourkouzelis and Peter Arberg},
    title =     {{Accommodating a Maximum Transit Unit/Maximum Receive Unit (MTU/MRU) Greater Than 1492 in the Point-to-Point Protocol over Ethernet (PPPoE)}},
    pagetotal = 9,
    year =      2006,
    month =     sep,
}

@misc{NS3,
    
title={{NS3 Network Simulator}},
    year={2025},
    howpublished={\url{https://www.nsnam.org/}},
    note={Accessed: 2026, January}
}

@article{5G_LENA,
    title = {An {E2E} simulator for {5G} {NR} networks},
    journal = {Simulation Modelling Practice and Theory},
    volume = {96},
    pages = {101933},
    year = {2019},
    issn = {1569-190X},
    author = {Natale Patriciello and Sandra Lagen and Biljana Bojovic and Lorenza Giupponi}
}

@misc{FCC_Thresholds,
    author={FCC},
    title={{FCC Regulations: 47 CFR 96.39}},
    year={2025},
    howpublished={\url{https://www.ecfr.gov/current/title-47/chapter-I/subchapter-D/part-96/subpart-E/section-96.39}},
    note={Accessed: 2025, December}
}

@ARTICLE{9321158,
  author={Gür, Gürkan},
  journal={Journal of Communications and Networks}, 
  title={Expansive networks: Exploiting spectrum sharing for capacity boost and 6G vision}, 
  year={2020},
  volume={22},
  number={6},
  pages={444-454},
  doi={10.23919/JCN.2020.000037}}

@article{PATIL2024110697,
title = {A comprehensive survey on spectrum sharing techniques for {5G/B5G} intelligent wireless networks: Opportunities, challenges and future research directions},
journal = {Computer Networks},
volume = {253},
year = {2024},
issn = {1389-1286},
author = {Anita Patil and Sridhar Iyer and Onel L.A. López and Rahul J. Pandya and Krishna Pai and Anshuman Kalla and Rakhee Kallimani}
}

\end{document}